\documentclass[final]{IEEEtran}
\usepackage[latin1]{inputenc}
\usepackage{amsmath}
\usepackage{amsfonts}
\usepackage{amssymb}
\usepackage{amsthm}
\usepackage[final]{graphicx}
\usepackage{algorithm,algpseudocode}
\usepackage{multirow}
\usepackage{subfigure}
\usepackage{cite}

\newtheorem{theorem}{Theorem}[section]
\newtheorem{lemma}[theorem]{Lemma}
\newtheorem{proposition}[theorem]{Proposition}

\begin{document}

\author{Diego Valsesia, Enrico Magli\thanks{This work was supported by the European Space Agency (ESA-ESTEC) under grant 107104. The authors are with the Department of Electronics and Telecommunications, Politecnico di Torino, C.so Duca degli Abruzzi, 24, 10129 Torino, Italy. Email: firstname.lastname@polito.it}}
\title{A Novel Rate Control Algorithm for Onboard Predictive Coding of Multispectral and Hyperspectral Images}


\maketitle

\begin{abstract}
Predictive coding is attractive for compression onboard of spacecrafts thanks to its low computational complexity, modest memory requirements and the ability to accurately control quality on a pixel-by-pixel basis. Traditionally, predictive compression focused on the lossless and near-lossless modes of operation where the maximum error can be bounded but the rate of the compressed image is variable. Rate control is considered a challenging problem for predictive encoders due to the dependencies between quantization and prediction in the feedback loop, and the lack of a signal representation that packs the signal's energy into few coefficients. In this paper, we show that it is possible to design a rate control scheme intended for onboard implementation. In particular, we propose a general framework to select quantizers in each spatial and spectral region of an image so as to achieve the desired target rate while minimizing distortion. The rate control algorithm allows to achieve lossy, near-lossless compression, and any in-between type of compression, \emph{e.g.}, lossy compression with a near-lossless constraint. While this framework is independent of the specific predictor used, in order to show its performance, in this paper we tailor it to the predictor adopted by the CCSDS-123 lossless compression standard, obtaining an extension that allows to perform lossless, near-lossless and lossy compression in a single package. We show that the rate controller has excellent performance in terms of accuracy in the output rate, rate-distortion characteristics and is extremely competitive with respect to state-of-the-art transform coding.
\end{abstract}

\section{Introduction}
Image spectrometers collect vast amounts of data which can be used for a variety of tasks. Possible applications include geological research, terrain analysis, material identification, military surveillance and many others. Fine spectral resolution can be a desired featured when it comes to detecting fingerprints in the spectral response of a scene. Such applications are enabled by the richness of data captured by multispectral and hyperspectral sensors. A problem of handling such wealth of information naturally arises and calls for the use of compression methods. 

Algorithms to compress hyperspectral and multispectral images have been studied for a long time and are still an active subject of research. Onboard compression enables spacecrafts to save transmission time, allowing more images to be sent to the ground stations. The design of compression algorithms for onboard applications must carefully meet the limited resources in terms of computational power and memory available on the spacecrafts. Two main compression techniques are available in this scenario: transform coding and predictive coding. 

Transform coding relies on computing a linear transform of the data to achieve energy compaction and hence transmit few carefully chosen transform coefficients. One of the most popular approaches is JPEG2000 \cite{taubmanjpeg2000} and its multidimensional extension \cite{jpeg2000ext}. A wavelet-based 2D lossless and lossy compression algorithm has also been standardized for space applications \cite{ccsds122}. Spectral transforms to eliminate the inter-band redundancy have been subject of intense research. There exists an optimal transform for Gaussian sources, \emph{i.e.}, the Karhunen-Lo\`{e}ve transform (KLT) but its complexity does not match the computational resources typically available for onboard compression. Hence, low-complexity approximations to the KLT have been derived, such as the Pairwise Orthogonal Transform (POT) \cite{pot}, the fast approximate KLT (AKLT) \cite{aklt} and the AKLT$_2$ \cite{aklt2}. Transform coding allows to perform lossless and lossy compression and to accurately control the rate in a simple manner thanks to the simple relation between rate and quantized transform coefficients \cite{taubmanjpeg2000} \cite{rhodomain}. On the other hand, per-pixel quality control as in near-lossless compression is hard to obtain. A near-lossless layer can be added to a transform coder, \emph{e.g.}, as in \cite{carvajalmagli}, but this requires to also implement a decoder onboard. Transform coding also typically suffers from the problem of dynamic range expansion, which is a direct consequence of energy compaction. While it is difficult to generalize due to the availability of many different transforms and predictors, a transform generally uses many (past and future) pixels of the image to represent a given pixel, while a predictor generally employs few pixels in a causal neighborhood, thus making it less prone to performance loss when the prediction is reset over different image areas, \emph{e.g.}, in order to achieve error resilience.

Predictive coding uses a mathematical model to predict pixel values and encode only the prediction error. Adaptive linear prediction is often used \cite{calibrationartifacts, abrardoDSC, rizzolow, acap, asap, lastri} (\emph{e.g.}, the predictor considered in Sec.\ref{sec:123extension} relies on the LMS filter \cite{lms}, with the sign algorithm \cite{chosign} for weight update), but other methods have been devised as well, \emph{e.g.}, based on edge detection \cite{edgeprediction} or vector quantization \cite{vectorquant}. In lossless compression, the prediction residuals are written in the compressed file after entropy coding. Lossy compression instead quantizes them before entropy coding. The quantization step size determines the amount of compression and hence information losses with respect to the original image. Near-lossless compression is readily implemented by setting a maximum quantization step size, so that the quantization error never exceeds half of it.
On the other hand, rate control in a predictive coder is challenging because: i) no simple mathematical relationship between the rate and the quantized prediction residual exists, ii) the quality of the prediction, hence the magnitude of the residuals, and ultimately the rate depend on how coarse the quantization is; as an example, the analysis of quantizer error propagation in the feedback loop is considered in \cite{aiazzipyramid} for the case of Laplacian pyramids. These aspects are further discussed in Sec. \ref{sec:bkg}. 

In this paper we propose an innovative design of a rate controller for a predictive encoder.  We show that the proposed method can achieve accurate control, while having complexity suitable for onboard implementation. In particular, the algorithm is designed to work in line-based acquisition mode, as this is the most typical setup of spectral imaging systems. We first describe the proposed algorithm in general terms, as it can be applied to any predictive coder. Next, we focus our attention on using it with the LMS predictor used in the CCSDS-123 standard for lossless compression \cite{ccsds123}, which is an improved version of the Fast Lossless algorithm \cite{FL}. The resulting system can be seen as an extension of the standard featuring lossless, near-lossless and rate-controlled lossy compression. The rate controller provides lossy reconstructions with increasingly better quality, up to lossless encoding, as the target rate approaches that of lossless compression. Finally, the controller can also work in a hybrid rate-controlled and near-lossless mode by specifying the maximum quantization step size that the controller is allowed to use.  

The paper is organized as follows: in section \ref{sec:bkg} we review the literature on rate control methods; in section \ref{sec:basics} we outline the main idea and the basic steps involved in the algorithm; in section \ref{sec:blocktypes} we describe the specific steps of the algorithm; in section \ref{sec:modeB} we introduce a second version of the algorithm, achieving a more accurate control by introducing a slice-by-slice feedback mechanism exploiting the measured rate of the previously encoded slice; section \ref{hybrid} shows how the proposed rate controller can actually achieve control of both the rate and the maximum distortion, enabling a hybrid near-lossless rate control mode; section \ref{sec:123extension} proposes an extension of the CCSDS-123 standard to near-lossless and rate-controlled lossy compression; finally, in section \ref{results} we show the performance of the rate-control algorithm on some test images and compare the proposed extension of CCSDS-123 with state-of-the-art transform coding techniques.

\section{Background}
\label{sec:bkg}
Rate control is a relatively well studied problem in the field of image and video coding, where it fits the framework of rate-distortion (RD) theory. The main task of rate-distortion optimization methods is to minimize the distortion of an encoded source (\emph{e.g.}, an image or a video sequence) subject to a constraint on the rate. This problem of carefully allocating the available resources is typically tackled by means of two techniques: Lagrangian optimization and dynamic programming. A more comprehensive review of such methods is covered by \cite{OrtegaSPM}. 

The classical method of Lagrangian optimization was introduced by Everett \cite{everett}, and relies on defining a cost function using a Lagrange multiplier to trade off rate and distortion. In particular, assume we have a budget-constrained allocation problem, such as our rate control problem:
\begin{align}
\label{constr}
\underset{x(i)}{\mathsf{minimize~}} \sum_{i=1}^N D_{i,x(i)} \mathsf{~~subject~to~~}\sum_{i=1}^N R_{i,x(i)} \leq R_{\mathsf{target}}
\end{align}
where $x(i)$ is the coding option for unit $i$.
The Lagrangian formulation consists in the unconstrained minimization of $J_i = D_i + \lambda R_i$, which can be shown \cite{everett}\cite{Shoham_EfficientBitAllocation} to yield the same solution of \eqref{constr} for a suitable $\lambda = \lambda^*$. Furthermore, if the coding units are independent, the minimization can be carried out independently for each unit $i$. One of the main issues of this method is to find the appropriate value of $\lambda$ needed to find the optimal distortion, while satisfying the rate constraint. By noticing that $\lambda=\lambda(R)$ is a monotonic function of the rate, an iterative search strategy, such as dychotomic search, can be used to find the correct value of $\lambda$.

It is often the case that the coding units exhibit some form of dependency among each other, so that the coding decisions taken for one unit may have some impact on the other units. This is notably true for prediction-based systems \cite{shapiro_optimal}, where quantization of residuals introduces noise in the prediction loop and may degrade the quality of future predictions. The interdependency of the coding choices makes this problem more difficult to tackle and classical solutions, based on dynamic programming, typically model the dependencies using a tree or a trellis \cite{ortega_trellis} \cite{ortega_bitallocation} and find the optimal path by using the Dijkstra's shortest path algorithm \cite{shortestpath} or the Viterbi algorithm \cite{viterbi}. The rate constraint can be handled by a suitable pruning of the tree.

In this paper we are studying a problem of rate control in the context of predictive coding on board of spacecrafts, posing significant constraints on the complexity of algorithms that can be used. The previously cited methods all exhibit a complexity that is unsuitable for the scenario we are considering or are largely inefficient (\emph{e.g.}, the standard Lagrangian approach with i.i.d. assumptions only). Onboard rate control is performed very easily in the case of systems adopting transform coding \cite{satellite_data_compression}, \emph{e.g.}, wavelet-based methods. This is due to the possibility to use an i.i.d. assumption among different coding units, allowing to establish simple relationships between rate and quantized transform coefficients \cite{taubmanjpeg2000} \cite{rhodomain}. However, such models do not hold in the case of predictive compression, making our task harder. Our approach uses models and independence assumptions to simplify the problem but we are forced to introduce corrections to the output of the models due to the inevitable dependencies introduced by the propagation of errors in the feedback loop. While the proposed procedure does not generally yield the optimal solution, it is a practical algorithm that can be used in low-complexity scenarios, such as onboard compression; moreover, it indeed achieves almost optimal performance, as will be shown in Sec.\ref{RDnumericalperformance}.

\section{Rate Control Algorithm}
\label{sec:basics}
This section outlines the framework and the basic operations performed by the proposed rate control algorithm. 

The main idea behind the algorithm is to adopt a model to predict the rate and the distortion of the quantized prediction residuals. In order to achieve a flexible scheme allowing an effective control of the rate for various kinds of images, ranging from hyperspectral images (lots of bands, but typically small spatial resolution) to multispectral images (large spatial resolution, but few bands), the algorithm partitions the image into blocks of size $BS_x \times BS_y$, where $BS_x$ and $BS_y$ are tunable parameters. Partitioning into blocks allows to deal with the non-stationary behaviour of the image. In fact, the prediction mechanism is indeed able to eliminate slowly varying features but sudden variations in the image content (\emph{e.g.} discontinuities) are hardly predicted by the encoder, and consequently imply non-stationary prediction residuals. 

The task of the rate control algorithm is then to assign a quantization step to each block of residuals in a given spectral channel, according to the specified target rate. At the same time, this assignment affects the overall distortion introduced by the encoder and, hence, it should be chosen to keep the distortion as low as possible.
In this scenario, computational complexity plays a major role in many ways. First of all, typical memory capabilities of systems for onboard image compression allow the storage of a limited number of lines of the image with all their spectral channels. To match this limitation, the rate control algorithm operates one slice at a time, where we denote as \emph{slice} a structure composed of one row of blocks with all their spectral channels. Moreover, as will be explained in section \ref{sec:models}, the algorithm does not even need to store all the lines in the slice but just a few of them, thus requiring very little memory.

The main steps involved in the algorithm are:
\begin{enumerate}
\item the estimation of the variance of the unquantized prediction residuals by running the lossless predictor for a small number of lines (Sec. \ref{sec:models});
\item the $l_1$ projection algorithm to get an initial allocation of the quantization steps (Sec. \ref{sec:l1});
\item the Selective Diet algorithm for rate and distortion refinement (Sec. \ref{sec:SelectiveDiet}). 
\end{enumerate} 

\vspace*{-0.1cm}
\subsection{Rate and distortion models}
\label{sec:models}
We now introduce the model used to describe the prediction residuals in each block. This model allows to obtain closed-form expressions for the rate and the distortion of the quantized residuals in the block. 
It is commonly observed that accurate predictors tend to yield residuals with leptokurtic (high kurtosis) distribution, hence similar to the Laplace probability density function, which we use to model the distribution of prediction residuals: 
\begin{align}
f_r(x) = \frac{\Lambda}{2}e^{-\Lambda\vert x \vert},
\end{align} 
where $\Lambda$ is related to the variance $\sigma^2$ of the distribution by $\Lambda = \sqrt{\frac{2}{\sigma^2}}$.

We assume that the residuals in each block and the blocks themselves are independent of each other. This is a simplifying assumption in two ways. First, the prediction mechanism may fail to remove all the correlation among the residuals. However, this does not pose a significant problem as we expect that most of the correlation is removed, hence making our independence assumption very close to reality; the same assumption is made in rate allocation for transform coding, where transform coefficients are often assumed to be independent. Second and more important, the quantization of the residuals introduces noise that propagates in the prediction loop. This leads to dependencies among the residuals and among blocks. Optimizing the allocation of the quantization step sizes taking into account these dependencies can lead to improvements as the model becomes more accurate. However, one must resort to dynamic programming methods (\emph{e.g.}, the Viterbi algorithm) that would be far too complex for our scenario. Consequently, we have explored a simplified way of including the effect of quantization noise in our model, \emph{i.e.}, augmenting the variance of the block by an estimate of the noise variance, which corresponds to assuming that the residuals and the quantization noise are independent:
\begin{align}
\label{variance_estimation}
\tilde{\sigma^2} = \sigma^2 + \frac{Q^2}{12} , 
\end{align}
where $Q$ is the quantization step used in the same block in previous slice. We do this because the quantization step size of the current slice is not known when we need to use this model, as it is indeed the output of the rate control process. It can be noticed that $\frac{Q^2}{12}$ is the mean square error produced by uniform scalar quantization of step size $Q$ under the high-rate approximation.

The rate (expressed in bits-per-pixel) is derived as the entropy of an i.i.d. continuous source with Laplace distribution, after quantization by means of a uniform scalar quantizer with step size $Q$:
\begin{align}
\label{entropy}
R = -p_0 \log_2 p_0 - 2\sum_{i=1}^\infty p_i \log_2 p_i ~,
\end{align}
so we need the probability $p_0$ that the residual is quantized to the zero value and the probability $p_i$ of being mapped to the (positive) integer $i$. For the uniform scalar quantizer we can write:
\begin{align}
\label{p0}
p_0 = \int_{-\frac{Q}{2}}^{\frac{Q}{2}} \frac{\Lambda}{2}e^{-\Lambda \vert x \vert} dx = 1-e^{-\Lambda\frac{Q}{2}}
\end{align}   
\begin{align}
\label{pi}
p_i = \int_{iQ-\frac{Q}{2}}^{iQ+\frac{Q}{2}} \frac{\Lambda}{2}e^{-\Lambda \vert x \vert} dx = \frac{1}{2}\left( e^{-\Lambda \left(iQ-\frac{Q}{2}\right)} - e^{-\Lambda \left(iQ+\frac{Q}{2}\right)} \right)
\end{align}
Inserting \eqref{p0} and \eqref{pi} into \eqref{entropy}, it is possible to derive \eqref{rate_usq}.
\begin{figure*}[!t]
\begin{align}
\label{rate_usq}
R(\Lambda,Q) = &-\left( 1 - e^{-\Lambda \frac{Q}{2}} \right) \log_2 \left( 1-e^{-\Lambda \frac{Q}{2}} \right) - \frac{e^{-\Lambda \frac{Q}{2}}}{\log(2)} \left[\log\left( \frac{1-e^{-\Lambda Q}}{2} \right) +\frac{\Lambda Q}{2} - \frac{\Lambda Q}{\left( 1-e^{-\Lambda Q } \right) } \right]
\end{align}
\begin{align}
\label{dist_usq}
D(\Lambda,Q) = \frac{2-\frac{1}{4}e^{-\Lambda \frac{Q}{2}} \left( \Lambda^2 Q^2 + 4\Lambda Q + 8 \right)}{\Lambda^2}&+ \frac{-\Lambda Q \left( \Lambda Q +4 \right) + e^{\Lambda Q} \left[ \Lambda Q \left( \Lambda Q -4 \right) +8 \right] -8}{4\Lambda^2} \frac{e^{-\frac{3}{2}\Lambda Q}}{1-e^{-\Lambda Q}}
\end{align}
\hrulefill
\vspace*{4pt}
\end{figure*}

We use mean squared error (MSE) as distortion metric, which can be computed as
\begin{align*}
D(\Lambda,Q) &= \int_{-\frac{Q}{2}}^{\frac{Q}{2}} x^2 \frac{\Lambda}{2}e^{-\Lambda \vert x \vert} dx \\ &+ 2\sum_{i=1}^\infty \int_{iQ-\frac{Q}{2}}^{iQ+\frac{Q}{2}}  \left( x-iQ \right)^2 \frac{\Lambda}{2}e^{-\Lambda \vert x \vert} dx ,
\end{align*}
thus obtaining \eqref{dist_usq}.

We can notice that both the rate and the distortion are functions of the variance $\sigma^2$ of the unquantized residuals in the block and of the quantization step size $Q$, whose value is yet unknown. Each block in the slice has its own variance parameter and quantizations step size. The variance must be estimated, while obtaining the quantization step size is really the ultimate goal of the rate control algorithm. 
The variance can be estimated by running the predictor \emph{without quantizing the prediction residuals} for a certain number of lines. A small fraction of the total lines in the block are sufficient to get good estimates of the variance of the residuals. In a software implementation, this is one of the main factors impacting on final complexity because it requires to run the predictor essentially twice: the first time on a small subset of the lines, without quantization, to estimate variances and then, once the quantization steps have been calculated, to perform the actual encoding pass, quantizing the residuals.
The previous rate and distortion models are used by the algorithms presented in the following subsections to find the right value of $Q$ for each block to match the target rate globally and have a low distortion.

\subsection{Projection onto the positive $l_1$ ball}
\label{sec:l1}
\begin{figure}
\centerline{\includegraphics[width=0.8\columnwidth]{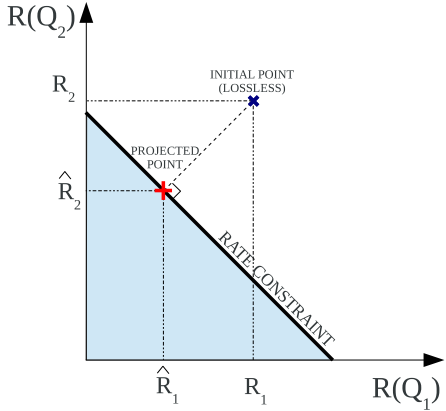}}
\vspace*{-0.3cm}
\caption{The rate point corresponding to the lossless allocation of Q's is projected onto the simplex defined by the rate constraint}
\label{simplex}
\end{figure} 

The goal of the algorithm described in the following is to provide an initial solution to the allocation problem. This solution, albeit inaccurate, is a good starting point to initialize the following algorithm (Selective Diet, explained in section \ref{sec:SelectiveDiet}). Suppose that the encoder is given a target rate for the encoded image equal to $T$ bits-per-pixel (bpp), and suppose that there are $N_B$ blocks in the current slice ($N_B$ is the product of the number of blocks in one band times the number of bands). We define the quantity $R_{\mathsf{target}} = T \cdot N_B$ as the product of the target rate in bpp and the number of blocks in the slice (note that this quantity does not represent the actual number of bits at our disposal since we are multiplying times the number of blocks and not the number of pixels). Ideally we would like to satisfy the rate constraint exactly, hence have 
\begin{align}
\label{rate_simplex}
\sum_{i=1}^{N_B} R(\Lambda_i,Q_i) = R_{\mathsf{target}}
\end{align}
where $Q_i$ is the quantization step size selected for the $i$-th block. Notice that since the rate of each block is a positive quantity, \eqref{rate_simplex} defines a simplex in $N_B$ dimensions. We can consider an initial solution having $Q_i=1 ~~ \forall i$ (lossless encoding), with corresponding rates $R(\Lambda_i,1)$. Geometrically (see Fig. \ref{simplex}), we have a vector in an $N_B$-dimensional space whose entries are the rates $R(\Lambda_i,1)$ and we can project it onto the simplex defined by \eqref{rate_simplex}. In other words, we seek to solve the following optimization problem, where we slightly abuse notation using boldface to indicate $N_B$-dimensional vectors and making the $R$ function operate component-wise:
\begin{align}
\label{projection}
\hat{\mathbf{R}} = \arg\min_{\mathbf{R}} \Vert \mathbf{R} - R(\mathbf{\Lambda},\mathbf{1}) \Vert_2 \mathsf{~~~subject~to~~~} \Vert \mathbf{R} \Vert_1 = R_{\mathsf{target}}
\end{align}
Problem \eqref{projection} is a continuous problem, whereas quantization step sizes are odd-integer-valued \footnote{Using odd-valued quantization step sizes is known to provide lower distortion for the same maximum error \cite{WuOdd}.}. After solving \eqref{projection} we need to search the value of $\hat{Q}_i$ such that $R(\Lambda_i,\hat{Q}_i)$ is closest to $\hat{R_i}$. Any search method such as linear search or binary search can be used for this purpose.

Projection onto a simplex is a special case of projection onto the $l_1$ ball, since the simplex is the positive part of the $l_1$ ball. $l_1$ projections algorithms have been subject of great interest in recent years due to surge in research on sparse methods. The field of compressed sensing \cite{donoho2006cs} has spawned from the discovery that $l_1$ penalized regressors can reconstruct a sparse signal exactly from a small number of random measurements, hence many reconstruction algorithms \cite{spgl1} include steps involving projections on the $l_1$ ball. We refer to the algorithm proposed in \cite{l1projections} to address the specific problem of projections onto the simplex. The algorithm has been shown to have $\mathcal{O}(N_B\log N_B)$ complexity. 
Being a continuous approximation to an integer-valued problem, the allocation returned by the projection algorithm can only provide a rough approximation to the desired rate. Nevertheless, it is expected to be close to a good solution, hence it is possible to improve it by making local modifications. This is the task performed by the Selective Diet algorithm.

\setlength{\textfloatsep}{10pt}{
\begin{algorithm}
\caption{Projection algorithm to solve \eqref{projection}}
\begin{algorithmic}
\State Sort $R(\mathbf{\Lambda},\mathbf{Q})$ into $\mathbf{\mu}$ in descending order
\State Find $\rho = \mathrm{max} \left\lbrace j : \mu_j - \frac{1}{j}\left( \sum_{i=1}^{N_B} \mu_i - R_{target} \right) > 0 \right\rbrace$
\State Define $\theta = \frac{1}{\rho} \left(  \sum_{i=1}^{N_B} \mu_i - R_{target} \right)$
\State Find $\mathbf{w}$ such that $w_i = \mathrm{max} \left\lbrace R(\Lambda_i,Q_i) - \theta ,0 \right\rbrace$
\State Find $\mathbf{\hat{Q}} = R^{-1}(\mathbf{\Lambda},\mathbf{w})$
\end{algorithmic}
\end{algorithm}}

\vspace*{-0.3cm}
\subsection{Selective Diet}
\label{sec:SelectiveDiet}
Selective Diet tries to solve an integer optimization problem consisting in lowering the distortion of the encoded slice while satisfying the constraint on its final rate. The algorithm is a local search method, similar in flavour to other discrete optimization methods such as hill climbing \cite{russellartificial} or meta-heuristics like tabu search \cite{tabu}.
At a high level it is possible to say that the algorithm is primarily concerned with finding a solution that meets the specification on the rate as closely as possible, while promoting solutions having low distortion. It does so by making local adjustments to the solution provided by the $l_1$ projector, hence the need for a good initialization point. A graphic visualization of a single iteration is shown in Fig. \ref{SDscheme}.

In this section, for convenience of explanation, we shall represent the blocks in the current slice as nodes in a chain. It is possible to modify the chain by making adjustments to the nodes, namely changing the quantization step size assigned to that node. Only local adjustments are allowed: the quantization step of each node can only be increased by 2 or decreased by 2. We shall call \emph{+2 level} an assignment of $Q_i+2$ where $Q_i$ is the current value of the quantization step, called \emph{default level}, and \emph{-2 level} an assignment equal to $Q_i-2$. A chain can be formed by choosing one of those three levels for each and every node. Consistently with the notation, we will call \emph{+2}/\emph{default}/\emph{-2} \emph{chain} a chain made only of nodes in the +2/default/-2 level. The ultimate goal of Selective Diet is creating a chain that meets the rate constraint and has low distortion. Let us now introduce a lemma at the basis of the local adjustments made to the default chain.

\begin{lemma}
\label{SDlemma}
Suppose that the default chain satisfies $\sum_{i=1}^{N_B} R(\Lambda_i,Q_i) = R_\mathsf{target}$, then if there exists a new chain satisfying $\sum_{j=1}^{N_B} R(\Lambda_j,Q_j) = R_\mathsf{target}$, it must contain nodes from both the $+2$ and $-2$ levels.
\end{lemma}

\begin{proof}
By contradiction, suppose that a chain meeting the rate constraint exists and is composed of nodes from the $+2$ and default levels only. However, $R(\Lambda_i,Q_i^{(+2)}) < R(\Lambda_i,Q_i^{(def)})$, so it must be that $\sum_{i=1}^{N_B} R(\Lambda_i,Q_i^{(ch)}) < R_\mathsf{target}$. Hence the rate is not met and such a chain does not exist.
Similarly, suppose that a chain meeting the rate exists and is composed of nodes from the $-2$ and default levels only. However, $R(\Lambda_i,Q_i^{(-2)}) > R(\Lambda_i,Q_i^{(def)})$, so it must be that $\sum_{i=1}^{N_B} R(\Lambda_i,Q_i^{(ch)}) > R_\mathsf{target}$. Hence the rate is not met and such a chain does not exist.
Therefore, a chain meeting the target can exist only if it uses nodes from both the $+2$ and $-2$ levels.
\end{proof}
Relying on this lemma, even when the rate is exact the algorithm must try to move some nodes to the -2 and +2 levels in order to optimize the distortion. The starting point is to consider the -2 chain as the new candidate output chain, since it has the lowest distortion. Obviously, selecting the -2 chain causes an increase in the rate, which must be compensated to meet the target. In order to reduce the rate moving back towards the target, some nodes are assigned to the +2 level. Each node is associated a cost function that considers the trade-off between the gain in rate reduction and the loss in quality due to switching from the -2 to the +2 level. The following cost function modelling the trade-off with a Lagrange multiplier is used:
\begin{align}
\label{SDcost}
J_i &= \left[ D(\Lambda_i,Q_i^{(-2)}) - D(\Lambda_i,Q_i^{(+2)}) \right] \notag \\ &+ \lambda \left[ R(\Lambda_i,Q_i^{(-2)}) - R(\Lambda_i,Q_i^{(+2)}) \right] ~~~~ i \in [1,N_B]
\end{align} 
The nodes are sorted by decreasing value of this cost function and this is the order in which the nodes are selected to be assigned to the +2 level. Specifically, one node at a time is added to the +2 level until the rate reaches $R_\mathrm{target}$. The new chain is then formed by the nodes that remained at the -2 level and the nodes that were demoted to the +2 level. This chain is taken as the new default chain for a new iteration of the algorithm in order to try to further improve distortion. Notice that even if in a single iteration the algorithm selects nodes from the +2 and -2 levels only, it is possible to reach any value of $Q$ using successive iterations, thus considering all possible odd values of the quantization step as possible choices for any block. The algorithm is run in a greedy manner, stopping when the distortion is not improving further. We have experimentally observed that the algorithm requires very few iterations (typically less than 10).

It should be noted that the $l_1$ projector may occasionally provide an initial solution that is not close enough to the target rate. We address this issue in the following way: if $\sum_{i=1}^{N_B} R(\Lambda_i,\hat{Q}_i) \leq 0.99 R_\mathrm{target}$, it means that the solution of the $l_1$ projector is underutilizing the available rate, so lemma \ref{SDlemma} does not hold and we run an iteration of Selective Diet with only the default and -2 chains. Instead, when the rate exceeds the target, running Selective Diet in the standard fashion already allows to reduce it back to the target, so no modification is made. Finally, the value of $\lambda$ controls the tradeoff between the reduction in rate and increase in distortion when adding a node to the +2 level. The optimal value of $\lambda$ would let us choose those nodes that allow a maximization of the gain in rate and a minimization of the increase in distortion. However, finding the optimal value would be computationally very demanding, so we resort to initializing $\lambda$ to an empirically determined value ($\lambda=50$) that we observed to be performing nicely over the whole test image set. This value is adjusted dynamically by the algorithm, halving it every time an increase in the overall distortion is observed in place of a decrease and rerunning the optimization with the new value. It is also possible to devise a lower complexity solution that does not adjust $\lambda$ and does not repeat the optimization procedure, at a price of lower performance.

The complete algorithm is summarized in Algorithm \ref{SDiet}.

\begin{figure}
\vspace*{-0.2cm}
\centerline{\includegraphics[width=0.9\columnwidth]{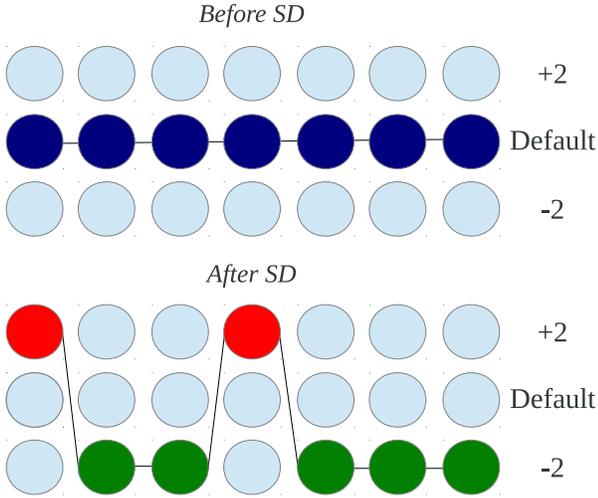}}
\caption{One iteration of Selective Diet tries to reduce the quantization step size by 2, but due to the increase in rate, the step size is actually increased by 2 for some blocks chosen as the best tradeoff between increase in distortion and gain in rate. Note that blocks of the default chain have different steps sizes, although the chain is depicted as a straight line for convenience.}
\label{SDscheme}
\end{figure} 

\begin{algorithm}
\caption{Selective Diet}
\begin{algorithmic}
\Require $\mathbf{Q_g}$, $\lambda=50$, $N_{iter}$
\For{$iter=1 \to N_{iter}$}
\State Set \textbf{default} = $\mathbf{Q_g}$ , {$\mathbf{Q}^{(+2)}$} = \textbf{Q+2}, $\mathbf{Q}^{(-2)}$ = \textbf{Q-2} 
\State Set output chain $\mathbf{Q_g} = \mathbf{Q}^{(-2)}$
\State Compute $R_{diff} = \sum R(\mathbf{Q_g}) - R_{target}$, \emph{i.e.}, the rate you need to lose to reach the target
\State Sort the nodes in {$\mathbf{Q}^{(+2)}$} by decreasing value of $J_i = \left( D_i^{(-2)} - D_i^{(+2)} \right) + \lambda \left( R_i^{(-2)} - R_i^{(+2)} \right)$
\State $i=1$
\While{$\sum R(\mathbf{Q_g}) - R_{target} < R_{diff}$}
\State Replace the corresponding node in $\mathbf{Q_g}$ with the i-th node in the sorted {$\mathbf{Q}^{(+2)}$}
\State $i=i+1$
\EndWhile  
\If{$iter \neq 1$}   
\If{Distortion did not lower AND inner iterations not exceeded}
\State Set $\lambda \gets \lambda/2$ and repeat current iteration
\Else
\State Proceed to next iteration
\EndIf
\EndIf
\EndFor
\end{algorithmic}
\label{SDiet}
\end{algorithm}

\section{Block Classification}
\label{sec:blocktypes}
The previous section outlined the basic operations of the rate control algorithm. We have discussed how models can be used to predict the rate and distortion of quantized blocks of prediction residuals. We have also introduced the $l_1$ projector and the Selective Diet algorithm that exploit the models to solve the problem of allocating quantization step sizes to the blocks to achieve the desired rate with low distortion. However, some improvements can be made in order to introduce additional features and solve problems not accounted for by the models; in this section we describe how blocks can be classified into three distinct classes to address those issues.

In particular, each block can be of one out of three types, labelled as: $\mathrm{NORMAL}$, $\mathrm{INFTY}$, $\mathrm{SKIP}$. The $\mathrm{NORMAL}$ type is for regular blocks not falling in any of the other categories, whose behaviour in the algorithm is just as described so far. 

The $\mathrm{INFTY}$ type is for blocks that are estimated to have a very low variance of the prediction residuals (\emph{e.g.}, $\sigma^2 < 0.1$). This happens for blocks in which the original image is very uniform so that most of the residuals are zero or close to zero. The rate spent for these blocks is mostly determined by quantization noise in prediction loop, but this is not detected during variance estimation because it is run in a lossless fashion, thus not producing any quantization noise. This means that the simplifying assumption of \eqref{variance_estimation} does not hold. Underestimating the variance will result in very inaccurate estimates of the rate of those blocks and improper allocation of the quantization steps, potentially affecting other blocks due to the propagation of quantization errors. Therefore, in the algorithm we exclude $\mathrm{INFTY}$ blocks from the projection and Selective Diet steps in order to avoid feeding those algorithms with misleading information. $\mathrm{INFTY}$ blocks are then treated separately. After the projector returns its initial solution, whenever an $\mathrm{INFTY}$ block is encountered in the slice, the same $Q$ as the closest $\mathrm{NORMAL}$ block in the same band is assigned to it. If no $\mathrm{NORMAL}$ block has been encountered yet and it is not the first slice then the same $Q$ of the block in the same position in the previous slice is used. Otherwise, if it is the first slice, $Q = 1$ is used. If the last encountered block is not a $\mathrm{NORMAL}$ block but a $\mathrm{SKIP}$ block, then the current $\mathrm{INFTY}$ block becomes a $\mathrm{SKIP}$ block. Except when the block becomes $\mathrm{SKIP}$, the target rate is updated for the Selective Diet algorithm. It is assumed that the $\mathrm{INFTY}$ block is driven by quantization noise so the target rate is updated as
\begin{align}
R_\mathrm{target} \leftarrow R_\mathrm{target} - R\left(\sqrt{\frac{24}{Q^2}},Q\right)
\end{align}

Optionally, $\mathrm{SKIP}$ blocks can be generated as a way to perform a further rate-distortion optimization by deciding to ``skip" a block, \emph{i.e.}, set to zero the prediction residual for all samples in the block and signal it using a 1-bit flag, if the predicted increase in distortion is low compared to the rate saving obtained by not encoding the block at all. However, skipping may introduce significant noise in the prediction loop, so the amount of skipped blocks must be controlled. Block skipping is useful only at low rates, therefore $\mathrm{SKIP}$ blocks can be generated only when the target rate is below 1 bpp, and a fixed percentage of blocks is skipped, as function of the target rate. This percentage increases as the target rate decreases
according to the following rule:
\begin{align}
p_s = \begin{cases} (1 - T)^3 &\mbox{if }T \leq 1\\ 
0 &\mbox{otherwise }
\end{cases}  
\end{align}
In order to choose which blocks must be skipped, the blocks in the current slice are sorted by decreasing value of $\Lambda$ and the first blocks in the sorted order will
be skipped. 

\section{Feedback-based mode}
\label{sec:modeB}

The rate control algorithm outlined so far is completely model-based, meaning that no information about the real rate of the encoded slices is available. We shall refer to this method as \textsc{MODE A} of the algorithm. A more accurate control can be achieved by adding a feedback mechanism that modifies the target rate for future slices based on the actual rate used to encode the previous slices. In particular, \textsc{MODE B} of the algorithm measures how many bits have been used to encode previous slices and adjusts the input target rate for the next slices so to achieve the global target rate. Note that we do not want to increase the complexity of the system, hence we are not performing a multi-pass encoding of the same slice but rather correcting the target for future slices. Although \textsc{MODE B} does not increase the complexity and can achieve more accurate control, it might lower the rate-distortion performance of the encoded image. To see this, let us consider a toy case in which the image is made of two slices, having the same rate-distortion function. The global rate-distortion curve for the whole image is convex, but, by adjusting the rate on a slice-by-slice basis, we operate on two distinct points of the curve and the final rate-distortion point lies on a straight line joining the two operating points, certainly above the convex curve. Hence, per-slice oscillations in the target rate introduce some suboptimality, which is more severe the farther apart the operating points of each slice lie in the rate-distortion plane.

\textsc{MODE B} adopts a Least-Mean-Square tracking approach to determine the target rate for the next slice, after measuring the rate produced by the encoding of the current slice. The target update formula is derived to take into account two issues. First, the inaccuracies in the rate controller make the actual output rate different from the target, thus we want to estimate the input-output relationship of the controller and track it in case of nonstationary behaviour. Second, we would like to count how many bits were used up to the current slice, and modify the target rate depending on the amount of bits that we saved, and we would like to spend on the next slices or, viceversa, the number of bits that we spent but we should have not. The goal is to try to assign all, but not more than the budget bits at our disposal, by spending them on or saving them from the remaining slices. The final rate update formula, to be motivated hereafter, is:
\begin{align}
\label{newtargetbpp}
T_{new}[n+1] &= \eta[n+1] + \frac{c[n+1]}{\tau}\cdot \frac{1}{\bar{w}[n]}
\end{align}
with
\vspace*{-0.1cm}
\begin{align}
c[n+1] &= \sum_{k=0}^n \left( T - y[k] \right) = c[n] + T - y[n] \\
\eta[n+1] &= \eta[n] + \bar{w}[n]\left[ T - y[n] + \frac{c[n]}{\tau} \right] \label{NTupdate} \\
\bar{w}[n] &= \frac{1}{\vert \mathcal{I} \vert} \sum_{k\in \mathcal{I}} w[k]  
\end{align}
where $y[n]$ is the actual rate produced encoding slice $n$, $T_{new}[n+1]$ is the target rate specified to the $(n+1)$-th slice, which is the next slice to be coded, and $T$ is the original target rate for the whole image (and the initial condition for $T_{new}$). $c[n]$, which we call ``residual budget'', stores how much deviation in rate from $T$ has been accumulated up to slice $n$. The $\tau$ factor used in the formulas plays the role of a time constant, ideally distributing the residual budget over $\tau$ future slices. It can be noticed that equation \eqref{costtrackingbudget} reduces to just a tracking term, when $\tau=+\infty$. Also notice that, for $\tau=2$, the residual budget term in \eqref{costtrackingbudget} is exactly $(c[n+1])^2$. $\bar{w}[n]$ is the ratio between output and input rate, averaged over the $\vert \mathcal{I}\vert$ previous slices identified by set $\mathcal{I}$, and $\vert \mathcal{I}\vert$ denotes the cardinality of the set. As we shall see, different choices of $\mathcal{I}$ are possible and yield different results. As special cases, we notice that, when $\mathcal{I}=\lbrace n \rbrace$, the algorithm does not average on previous slices, hence it is most suited for highly non-stationary scenarios, while, when $\mathcal{I}=\left\lbrace0,1,\ldots,n\right\rbrace$, the algorithm uses all the history for averaging, yielding the best performance for stationary scenarios. The following theorems prove the rate control performance in such special cases. The wide sense stationarity (WSS) assumption that we make in the proofs has been verified to be a rather good model, since it basically means that the non-ideal behaviour of the rate controller, that we are trying to correct, has certain regularity properties. We remark that experimental results showed that when $w[n]$ is WSS, the output rate of the memory-1 method converges to $T$ but the residual budget converges to a non-zero value proportional to the variance of $w[n]$. This is why we advocate that the long-memory method is better when we expect a stationary behaviour. However, the memory-1 method (or a method with a limited memory) is better for tracking non-stationarities thanks to Theorem \ref{costmin}. In the proofs we will denote $\frac{c[n+1]}{\tau}\cdot \frac{1}{\bar{w}[n]} =  \xi$ for brevity.

\begin{proposition}[\bf{Convergence of long-memory method}]
Let the rate controller obey the input-output relationship $y[n] = w[n]T_{new}[n]$, being $w[n]$ a wide sense stationary random process with mean $\mathbb{E}\left[ w[n] \right] = \mu$ and $\mathbb{E}\left[ \left( w[n+l]-w \right) \left( w[n]-w \right) \right] = 0 , \forall l\neq 0$. Let $T_{new}$ be updated as in \eqref{newtargetbpp} with $\mathcal{I}=\left\lbrace0,1,\ldots,n\right\rbrace$. Then,

\begin{align}
\lim_{n\rightarrow +\infty} \mathbb{E}\left[ y[n] \right] &= T \tag{Convergence to target} \\
\lim_{n\rightarrow +\infty} \mathbb{E}\left[ c[n] \right] &= 0 \tag{Convergence to zero residual budget}
\end{align}
\end{proposition}
\begin{proof}
We will not give a formal proof of this result, rather just a sketch. We notice that the sequence of averages over $n$ samples $\bar{w}[n]$ has a limit $\lim_{n\rightarrow \infty} \bar{w}[n] = \mathbb{E}\left[ w[n] \right] = \mu$ thanks to the ergodicity of $w[n]$. We suppose that it reaches this limit value fast and thus we approximate $\bar{w}[n] \approx \mu$ for all $n>n_0$. Using this fact and performing some algebraic manipulations on \eqref{NTupdate} and \eqref{newtargetbpp}, similar to those done in the proof of Theorem \ref{convergence}, we obtain the following recursion
\begin{align*}
T_{new}[n+1] &= \frac{\mu}{\tau}T + \left( 2 - \mu w[n] - \frac{w[n]}{\mu\tau} \right) T_{new}[n] \\ &- \left( 1-\mu w[n]-\frac{w[n]}{\mu\tau} + \frac{\mu^2}{\tau} \right) T_{new}[n-1]
\end{align*}
Hence 
\begin{align*}
\mathbb{E}\left[ y[n+1] \right] &= \mu \left( 2 - \mu^2 - \frac{1}{\tau} \right) \mathbb{E}\left[ T_{new}[n] \right] \\ &- \mu \left( 1-\mu^2-\frac{1}{\tau} + \frac{\mu^2}{\tau} \right) \mathbb{E}\left[ T_{new}[n] \right] + \frac{\mu^2}{\tau}T
\end{align*}
We take the limit on both sides to get
\begin{align*}
y^* &=  \left( 2 -\mu^2-\frac{1}{\tau} -1+\mu^2 + \frac{1}{\tau}  - \frac{\mu^2}{\tau} \right) y^* + \frac{\mu^2}{\tau}T \\
y^* &= \lim_{n\rightarrow +\infty} \mathbb{E}\left[ y[n] \right] = T.
\end{align*}
Similarly, the residual budget term can be shown to follow
\begin{align*}
c[n] &= \frac{1}{\frac{\mu}{\tau} + \frac{1}{\mu\tau}- \frac{1}{\mu\tau}} \cdot \\ & \Big[ T_{new}[n] - \left( 1-\mu w[n-1] - \frac{w[n-1]}{\mu\tau} \right) T_{new}[n-1] \\ &- \left(\mu+\frac{1}{w\tau} \right)T \Big] + T - w[n-1] T_{new}[n-1]
\end{align*}
Hence,
\begin{align*}
&\lim_{n\rightarrow +\infty} \mathbb{E}\left[ c[n] \right] = \frac{1}{\frac{\mu}{\tau} + \frac{1}{\mu\tau}- \frac{1}{\mu\tau}} \cdot \\ & \left[ \frac{T}{\mu} - \left( 1-\mu^2-\frac{1}{\tau} \right) \frac{T}{\mu} - \left( \mu+\frac{1}{\mu\tau} \right)T \right] + T - \mu\frac{T}{\mu} = 0
\end{align*}
\end{proof} 

\begin{theorem}[\bf{Convergence of memory-1 method}]
\label{convergence}
Let the rate controller obey the input-output relationship $y[n] = w\cdot T_{new}[n]$, with $w^2 \leq 2$ and let $T_{new}$ be updated as in \eqref{newtargetbpp} with $\mathcal{I}=\lbrace n \rbrace$. Then,
\begin{align}
\lim_{n\rightarrow +\infty} y[n]  &= T \tag{Convergence to target} \\
\lim_{n\rightarrow +\infty} c[n]  &= 0 \tag{Convergence to zero residual budget}
\end{align}
\end{theorem}
\begin{proof}
\begin{align}
\label{yproof}
y[n] &= w \left( \eta[n] + \frac{c[n]}{\tau w}  \right) \notag \\ &= \left( 1- w^2 - \frac{1}{\tau} \right)y[n-1] + w^2T + \frac{T}{\tau} + w^2\frac{c[n-1]}{\tau}
\end{align}
However, from the definition of $c[n]$:
\begin{align}
\label{cproof}
c[n] = c[n-1] + T - y[n-1]
\end{align}
We can solve \eqref{yproof} for $c[n-1]$ and insert it in \eqref{cproof}.
\begin{align}
\label{ccrucial}
c[n] &= \frac{\tau}{w^2} \left( y[n] - \left( 1-w^2-\frac{1}{\tau} \right)y[n-1] -w^2T - \frac{T}{\tau} \right) \notag \\ &+ T - y[n-1]
\end{align}
We can recall that
\begin{align*}
y[n+1] &= \left( 1- w^2 - \frac{1}{\tau} \right)y[n] + w^2T + \frac{T}{\tau} + w^2\frac{c[n]}{\tau} \\
&= \left( 2-w^2-\frac{1}{\tau} \right)y[n] \\&- \left( 1+\frac{1-\tau}{\tau}w^2 - \frac{1}{\tau} \right)y[n-1] + \frac{w^2}{\tau}T
\end{align*}
The general solution to that difference equation, considering the initial conditions $y[0]=wT$ and $y[1]=wT+w^2T-w^3T+\frac{T-wT}{\tau}$, is
\begin{align*}
y[n] &= \frac{T(1-w)}{\tau w^2-1} \left( 1-\frac{1}{\tau} \right)^n \\ &+ \left[ Tw - T - T \frac{1-w}{\tau w^2-1} \right] \left( 1-w^2 \right)^n + T
\end{align*}
It is easy to check the limit:
\begin{align*}
\lim_{n\rightarrow \infty} y[n] = T ,
\end{align*}
provided that $w^2 \leq 2$. Moreover we can take the limit of \eqref{ccrucial} to check budget convergence:
\begin{align*}
\lim_{n\rightarrow \infty} c[n] = \frac{\tau}{w^2} \left( T - \left( 1-w^2-\frac{1}{\tau} \right)T -w^2T - \frac{T}{\tau} \right) = 0
\end{align*}
\end{proof}

\begin{theorem}[\bf{Cost minimization of memory-1 method}]
\label{costmin}
Let the rate controller obey the input-output relationship $y[n] = w[n]T_{new}[n]$, and let $T_{new}$ be updated as in \eqref{newtargetbpp} with $\mathcal{I}=\lbrace n \rbrace$. Then, update \eqref{NTupdate} is a gradient descent step towards the minimization of 
\begin{align}
\label{costtrackingbudget}
J = \underbrace{ \Bigg( T - y[n] \Bigg)^2 }_{\textsc{tracking}} + \underbrace{ \Bigg( T - y[n] + \frac{2c[n]}{\tau} \Bigg)^2 }_{\textsc{budget}}
\end{align}
\end{theorem}
\begin{proof}
\begin{align*}
J &= T^2 + y^2[n] -2Ty[n] + 4\frac{c^2[n]}{\tau^2} + T^2 + y^2[n]\\
&- 4\frac{c[n]}{\tau}y[n] + 4\frac{c[n]}{\tau}T - 2Ty[n]\\
&= 2T^2 + 2w^2[n]\xi^2 - 2Tw[n]\xi + 4\frac{c^2[n]}{\tau^2}  - 4\frac{c[n]}{\tau}w[n]\xi\\
&+ 4\frac{c[n]}{\tau}T -2T\xi w[n] + 2w^2[n]\eta^2[n] + 4w^2[n]\eta[n]\xi\\
&- 2Tw[n]\eta[n] 4\frac{c[n]}{\tau}w[n] \eta[n] - 2T\eta[n]w[n]
\end{align*}
\begin{align*}
\frac{\mathrm{d}J}{\mathrm{d}(\eta[n])} &= 4w^2[n]\eta[n] + 4\xi w^2[n] \\ &- 2Tw[n] - 4\frac{c[n]}{\tau}w[n] - 2Tw[n]
\end{align*}
The gradient descent update equation is
\begin{align*}
\eta[n+1] &= \eta[n] -\alpha \frac{\mathrm{d}J}{\mathrm{d}(\eta[n])} \\ &= \eta[n] - 4\alpha w[n] \left( y[n] - T - \frac{c[n]}{\tau} \right)
\end{align*}
Thus, setting $\alpha=\frac{1}{4}$ we obtain \eqref{NTupdate}.
\end{proof}

\section{Hybrid near-lossless rate control}
\label{hybrid}
\begin{figure*}
\subfigure[No $\mathrm{CLIP}$]{
\includegraphics[width=0.32\textwidth]{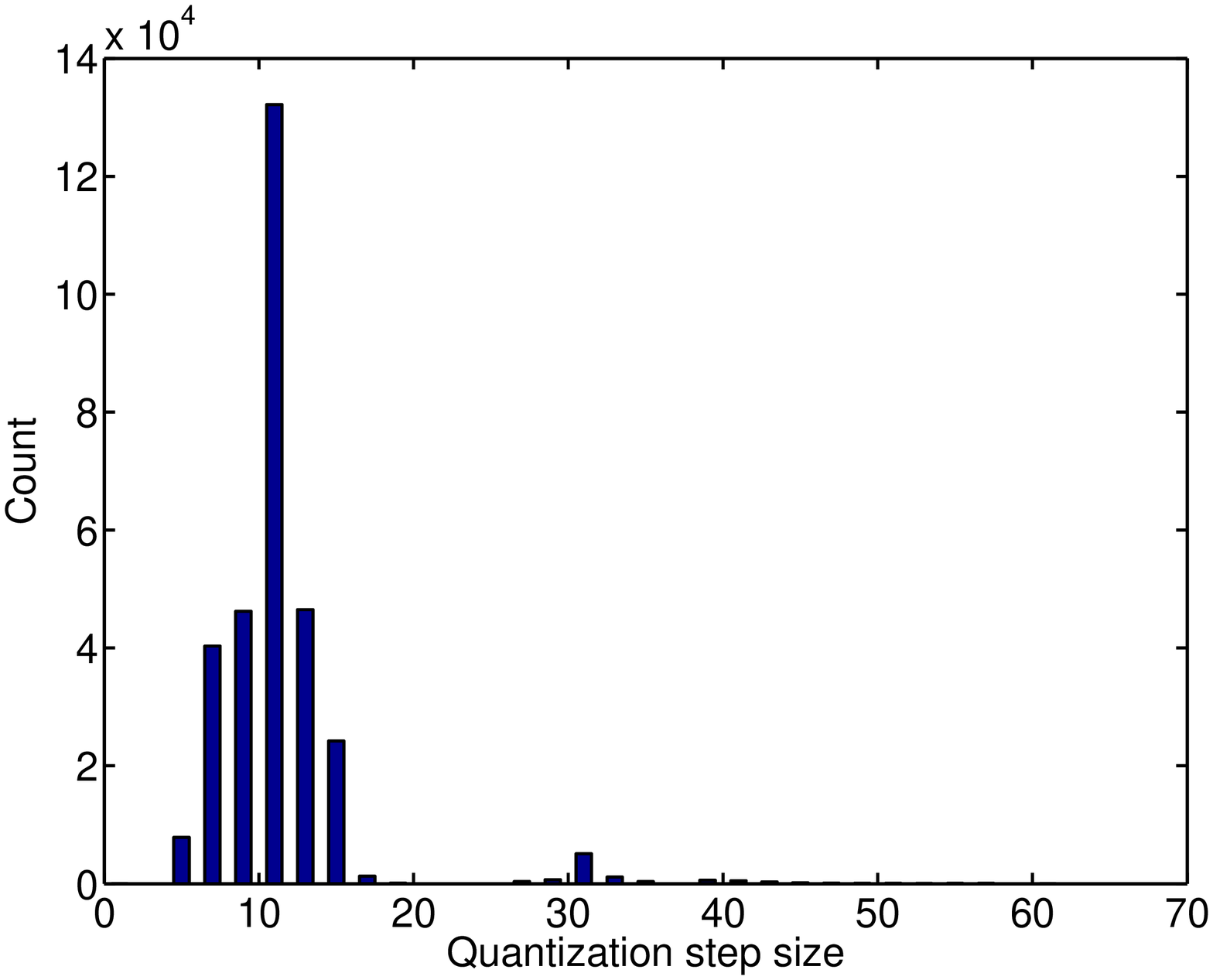}}
\subfigure[$\mathrm{CLIP}=21$]{
\includegraphics[width=0.32\textwidth]{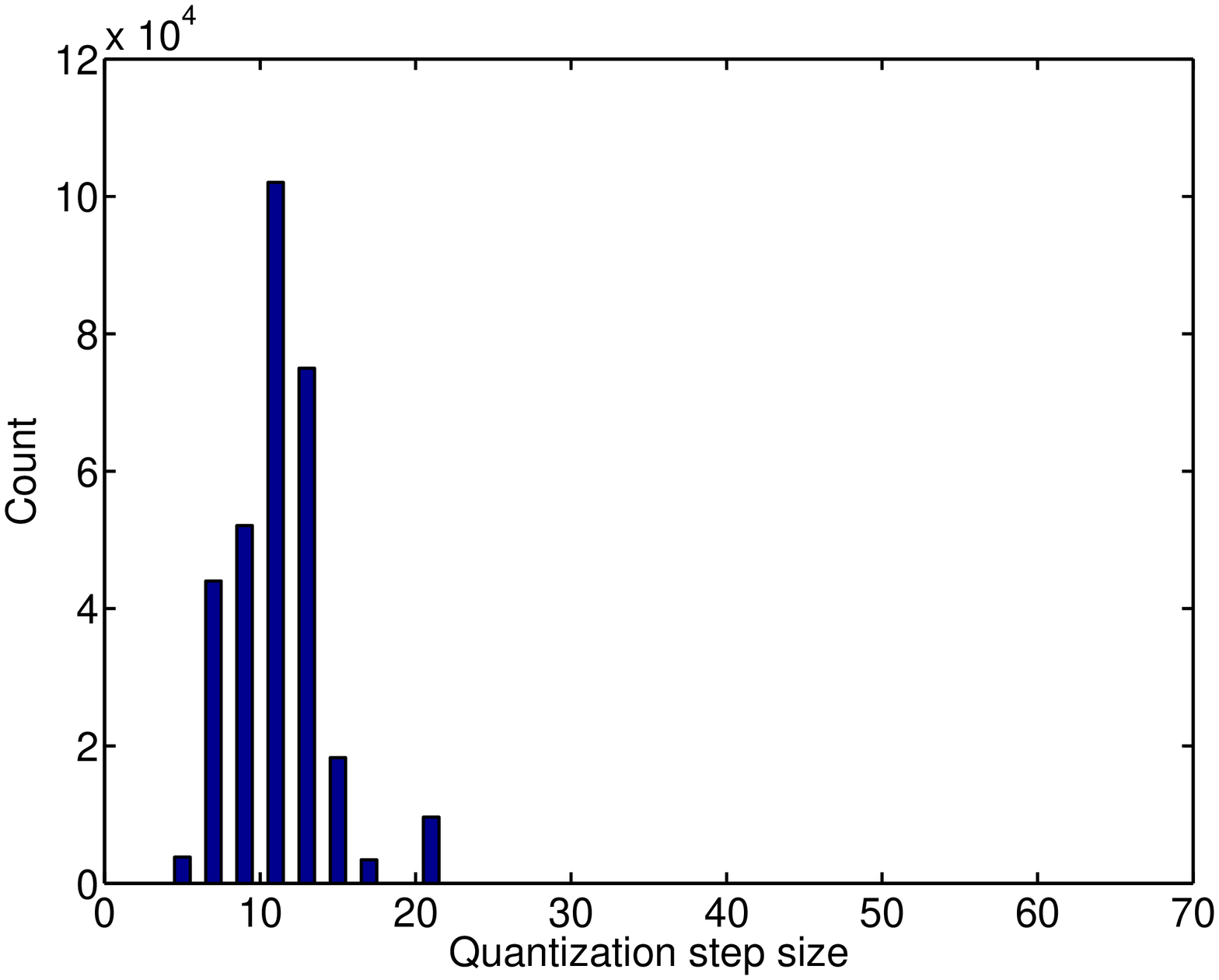}}
\subfigure[$\mathrm{CLIP}=11$]{
\includegraphics[width=0.32\textwidth]{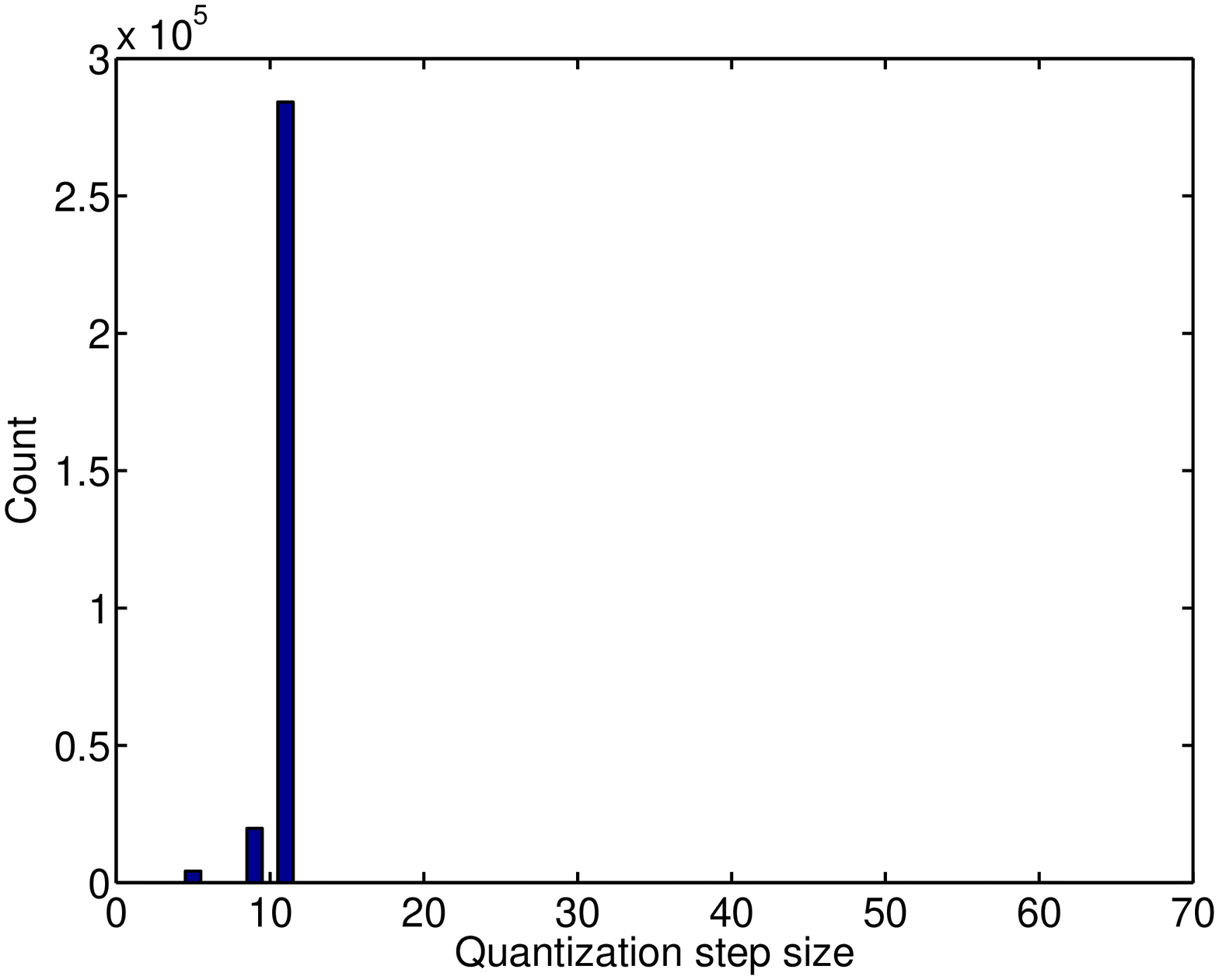}}
\caption{AVIRIS sc0\_raw. (a) rate$=3.0052$ bpp, MAD=30, SNR=62.25 dB; (b) rate$=3.0046$ bpp, MAD=10, SNR=62.85 dB; (c) rate$=2.9968$ bpp, MAD=5, SNR=63.39 dB}
\label{avirisclip}
\end{figure*} 
The proposed rate control algorithm opens the way for an interesting hybrid operating mode in which one can simultaneously constrain target rate and maximum distortion. This significantly differs from traditional operating modes in which one can either specify the rate but has no control over the per-pixel maximum error (as it typically happens in rate-controlled transform coding approaches) or in which one specifies the maximum error but has no control over the rate (as it is easily done in near-lossless predictive schemes). The implementation of such hybrid mode is trivial by using the proposed rate controller because it is sufficient to limit the maximum quantization step size allowed in the $l_1$ projector and in Selective Diet. If such specification is compatible with finding an allocation of quantization step sizes that yields the prescribed target rate, then the algorithm successfully controls both the rate and the maximum distortion.

Fig. \ref{avirisclip} reports the results of some experiments (see Sec. \ref{sec:123extension}-\ref{results} for more details on the test image) that graphically show the impact of constraining the maximum quantization step size (called $\mathrm{CLIP}$) on the distribution of quantization steps and on the rate and quality of the encoded image. In this case the controller successfully provides the desired rate even with the very demanding constraint $\mathrm{CLIP}=11$. Also, notice the improvement in terms of MAD and SNR obtained by the hydrid mode. The higher SNR obtained by enforcing a constraint on the maximum error should not be surprising. In fact, the $l_1$ projector and Selective Diet alone have no guarantee of optimality and enforcing an additional constraint allows to shrink the solution space, eliminating suboptimal allocations.  Finally, if the user were to demand $\mathrm{CLIP}=5$, she would actually get MAD=2 but the controller would be unable to provide the target rate of 3 bpp and, in fact, provides $4.0586$ bpp.

\section{Extension of CCSDS-123 to near-lossless and lossy compression with rate control}
\label{sec:123extension}
\subsection{Review of CCSDS-123}
The Consultative Committee for Space Data Systems (CCSDS) has recently developed the CCSDS-123 recommendation, intended for lossless compression of multispectral and hyperspectral images. CCSDS-123 is based on the Fast Lossless compression algorithm \cite{FL} \cite{calibrationartifacts}, which is a predictive method. The algorithm computes a local sum $\sigma_{z,y,x}$, obtained from a causal neighborhood of the pixel. A weighted combination of the local sums in the $P$ previous bands yields the predicted pixel value. The algorithm adapts the weights using the sign algorithm \cite{chosign}, which is a low-complexity solution for the implementation of a least-mean-square filter. 

Let $s_{z,y,x}$ denote the pixel value at position $(x,y,z)$, then the encoder computes:
\begin{align*}
\hat{d}_{z,y,x} = \mathbf{W}_{z,y,x}^T \mathbf{U}_{z,y,x} = \mathbf{W}_{z,y,x}^T \left[\begin{array}{c}
4s_{z,y-1,x}-\sigma_{z,y,x}\\
4s_{z,y,x-1}-\sigma_{z,y,x}\\
4s_{z,y-1,x-1}-\sigma_{z,y,x}\\
4s_{z-1,y,x}-\sigma_{z-1,y,x}\\
\vdots\\
4s_{z-P,y,x}-\sigma_{z-P,y,x}
\end{array}\right]
\end{align*}
A \emph{scaled predicted sample}  $\tilde{s}_{z,y,x}$ is calculated from $\hat{d}_{z,y,x}$. 
The prediction residual is computed as $\Delta_{z,y,x} = s_{z,y,x} - \left\lfloor \frac{\tilde{s}_{z,y,x}}{2} \right\rfloor$ and then mapped to a positive integer $\delta_{z,y,x}$ to be entropy encoded.
For further details, we refer the reader to the CCSDS-123 Blue Book \cite{ccsds123} and to the paper by Aug\'{e} \emph{et al.} \cite{ccsds123performance} for a more throughout explanation of the encoder parameters and their impact on performance.

\subsection{Near-lossless extension}
Extending the compression mechanism to near-lossless encoding simply requires to introduce a quantizer in the prediction loop. In particular, we use a uniform scalar quantizer to quantize the prediction residual $\Delta_{z,y,x}$ into $\hat{\Delta}_{z,y,x} = \mathrm{sgn}\left( \Delta_{z,y,x} \right)\cdot \left\lfloor  \frac{\vert \Delta_{z,y,x} \vert + (Q-1)/2}{Q} \right\rfloor$. The quantized value is then mapped to a positive integer and sent to the entropy coding stage. In order to have synchronization with the decoder, we must consider the dequantized value $Q \hat{\Delta}_{z,y,x}$ for weight update. The near-lossless encoder uses a single quantization step size for the whole image.

\subsection{Rate-controlled lossy extension}
The rate-controlled version of the algorithm uses the proposed rate control method to assign a different quantization step size to each block in the image. Assuming that the encoder proceeds in a Band Interleaved by Line (BIL) order, the rate control procedure is called whenever the current pixel belongs to the first band and it is at the beginning of a new slice (\emph{i.e.}, position $z=0$, $y=k\cdot BS$, $x=0$). As explained in the previous sections, the rate controller first tries to encode $\mathrm{ESTLINES}$ lines (with all their spectral bands) in a lossless mode in order to estimate the variance of the prediction residuals. Once the variance is estimated and the allocation of quantization steps is performed, the encoder backtracks to position $(0,k\cdot BS,0)$, discarding all the weight updates done in the meanwhile and starts the actual encoding pass of the slice. Similarly to the near-lossless mode, the encoder now computes the quantized prediction residuals $\hat{\Delta}_{z,y,x}$, but now employing the quantization steps calculated by the controller for each block.

It is important to notice that the chosen quantization step sizes must be written in the header of the compressed file for usage at the decoder side. In order to keep the overhead low we propose to use a differential encoding strategy adopting the Exp-Golomb code \cite{expgolomb}. Differential encoding amounts to encoding only differences between two successive quantization steps and, since they are expected to be close to each other, some compression is obtained. A simple universal code such as the Exp-Golomb code of order zero is then used to compress the differences.

Finally, formulas \eqref{rate_usq} and \eqref{dist_usq} can be implemented by means of lookup tables. It can be noticed that the rate depends only on $\Lambda Q$ and that the distortion can be rewritten as the product of a function of $\Lambda Q$ and $Q^2$. We have verified that two lookup tables of roughly $45000$ integer values each are sufficient to ensure the correct behavior of the algorithm. The values in the rate table can be represented using 14 bits per value, while the distortion values need 13 bits. The total memory occupation of the two tables is thus about 152 kB.

\vspace*{-0.25cm}
\subsection{Range encoder}
\label{sec:range}
The CCSDS-123 recommendation defines an adaptive coding approach using Golomb-Power-of-2 codes, mainly due to its low complexity and good performance, as well as the existence of an earlier standard (CCSDS 121.0-B \cite{ccsds121}) using the Rice coding algorithm, embedded in the block-adaptive mode.

We propose a different entropy coding stage based on the range coder \cite{martinrange}. The range coder is essentially a simplified arithmetic encoder. Such a block coder is needed in order to achieve rates lower than 1 bpp, as the minimum codeword length for the Golomb code is 1 bit. Moreover, a higher performance entropy coder improves the effectiveness of the rate controller, by limiting the suboptimality introduced at this stage. 
For efficiency reasons, the proposed range coder keeps four separate models for each band for the prediction residuals, as described in \cite{losslessLUT}.

\vspace*{-0.1cm}
\section{Numerical results}
\label{results}
We have performed extensive tests on images extracted from the corpus defined by the MHDC working group of the CCSDS for performance evaluation and testing of compression algorithms. A total of 47 images is used to generate the ensemble statistics, while for brevity we report numerical results for a smaller subset. The whole corpus comprises images of various nature, from ultraspectral images captured by IASI and AIRS sensors, through hyperspectral images captured by CASI, SFSI, AVIRIS and Hyperion sensors, to multispectral images captured by MODIS, Landsat, Vegetation, MSG, Pleiades and SPOT5 sensors. Table \ref{tab:testimages} reports details about the images used in the tests and the number of bands $P$ used for prediction. The images with the \textsc{NUC} suffix present Non-Uniformity Correction, \emph{i.e.}, a form of compensation of the different gains of the lines of the image, performed by means of a median filter, as described in \cite{pot}.

The tests have multiple goals. First, we want to analyze the accuracy of the rate control algorithm, assessing how close the actual rate of the compressed image is with respect to the specified target. Second, we study the rate-distortion performance of the algorithm by drawing the full rate-distortion curve in order to compare it against the rate-distortion curve obtained by the near-lossless version of the encoder. This is known to be the optimal quantization step selection for a Gaussian source, but does not provide rate control, although many rate-distortion points are indeed achievable. We use this curve as an upper performance bound in order to estimate how close the proposed rate control algorithm can get to the ideal solution. Finally, we compare the performance of the proposed extension of CCSDS-123 to lossy compression with rate control against a state-of-the-art transform coder intended for onboard compression.

\begin{table}
\caption{Test Images}
\label{tab:testimages}
\centering
\begin{tabular}{c c c c c}
Image & Rows & Columns & Bands & $P$ \\
\hline
\hline
\textsc{AVIRIS sc0\_raw} & 512 & 680 & 224 & 15 \\
\hline
\textsc{AIRS gran9} & 135 & 90 & 1501 & 10 \\
\hline
\textsc{CASI-t0477f06-nuc} & 1225 & 406 & 72 & 2\\ 
\hline
\textsc{CRISM-sc167-nuc} & 510 & 640 & 545 & 3 \\ 
\hline
\textsc{CRISM-sc182-nuc} & 450 & 320 & 545 & 3 \\ 
\hline
\textsc{CRISM-sc214-nuc} & 510 & 640 & 545 & 3 \\ 
\hline
\textsc{frt00009326\_07\_vnir} & 512 & 640 & 107 & 3 \\ 
\hline
\textsc{Geo\_Sample\_Flatfielded} & 1024 & 256 & 242 & 10 \\ 
\hline
\textsc{M3targetB-nuc} & 512 & 640 & 260 & 3 \\ 
\hline
\textsc{M3targetB} & 512 & 640 & 260 & 3 \\ 
\hline
\textsc{MODIS-MOD01\_250m} & 8120 & 5416 & 2 & 1 \\ 
\hline
\textsc{MODIS-MOD01\_500m} & 4060 & 2708 & 5 & 4 \\ 
\hline
\textsc{MODIS-MOD01day} & 2030 & 1354 & 14 & 2 \\ 
\hline
\textsc{MODIS-MOD01night} & 2030 & 1354 & 17 & 4 \\ 
\hline
\textsc{montpellier} & 224 & 2456 & 4 & 3 \\ 
\hline
\textsc{mountain} & 1024 & 1024 & 6 & 5 \\ 
\hline
\textsc{t0477f06\_raw} & 1225 & 406 & 72 & 2 \\ 
\hline
\textsc{toulouse\_spot5\_xs\_extract1} & 1024 & 1024 & 3 & 3 \\ 
\hline
\textsc{vgt\_1b} & 10080 & 1728 & 4 & 3 \\ 
\hline
\end{tabular}
\end{table}

\subsection{Complexity considerations}
Before presenting the experimental performance of the proposed algorithm, we analyze its computational complexity both theoretically and on a real implementation. 

The lossless version of the compression algorithm is quite similar to the CCSDS-123 recommendation, with the exception of entropy coding stage, now replaced by the range coder. Its complexity and the one of the near-lossless scheme are therefore just marginally higher than CCSDS-123. The rate control algorithm has three main sources of complexity:
\begin{itemize}
\item the estimation of the variance of unquantized prediction residuals
\item the $l_1$ projector
\item the Selective Diet optimization algorithm
\end{itemize}
We remarked in Sec. \ref{sec:l1} that the $l_1$ projector has complexity $\mathcal{O}(N_B\log N_B)$, essentially due to the sorting procedure. The Selective Diet algorithm also has a sorting step as the main source of complexity. After the blocks in the current slice are sorted according to the value of the cost function, a linear scan is performed to optimize the quantization step sizes. This basic operation is repeated for $N_{iter}$ iterations, hence with good approximation we can say that Selective Diet has $\mathcal{O}(N_{iter}(N_B\log N_B + N_B))$ complexity. However, it is typically observed that the number of required iterations is very low (around 5 to 10) and can be bounded to a predefined value.

We also profiled our C-language implementation of the compression algorithm and compared lossless encoding against rate-controlled encoding in terms of running times. We used the \emph{aviris sc0\_raw} image for our test, as it is one of the biggest in the dataset. Rate control was set to $3$ bpp with \textsc{MODE A}. The running time of the lossless encoder was $72.62$ seconds, while the rate-controlled encoder took $80.48$ seconds. The time spent writing to file was removed from both measurements in order to avoid any bias due to different file sizes. It can be noticed that the overhead of the rate controller is around $10\%$. Careful profiling of the code suggests that this overhead is due for $65\%$ ($5.11$ sec.) to variance estimation, while only $22\%$ ($1.73$ sec.) to optimization ($l_1$ projector and Selective Diet). The remaining $13\%$ is due to other inefficiencies in the code, which is not very optimized. This result confirms our intuition, presented in Sec. \ref{sec:models}, that variance estimation is the main source of complexity, and so the number of lines (\textsc{ESTLINES}) used for this task must be chosen carefully. All the results presented in this paper were obtained with \textsc{ESTLINES}$=2$.

\subsection{Accuracy of rate control}

\begin{figure*}[ht!]
\addtolength{\subfigcapskip}{-0.2cm}
\vspace*{-0.1cm}
\centering
\subfigure[1 bpp]{
\includegraphics[width=0.3\textwidth]{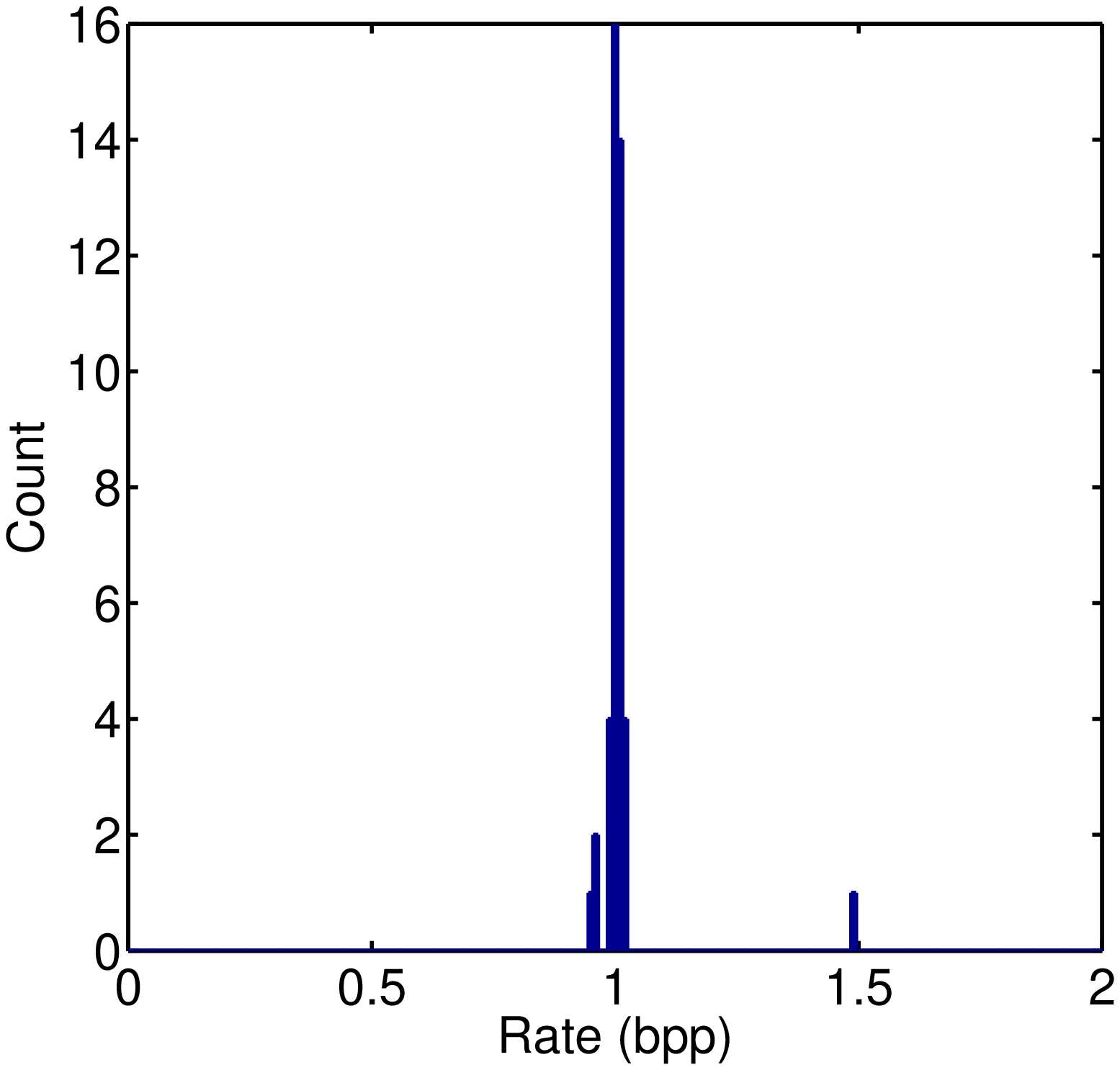}} 
\subfigure[2 bpp]{
\includegraphics[width=0.3\textwidth]{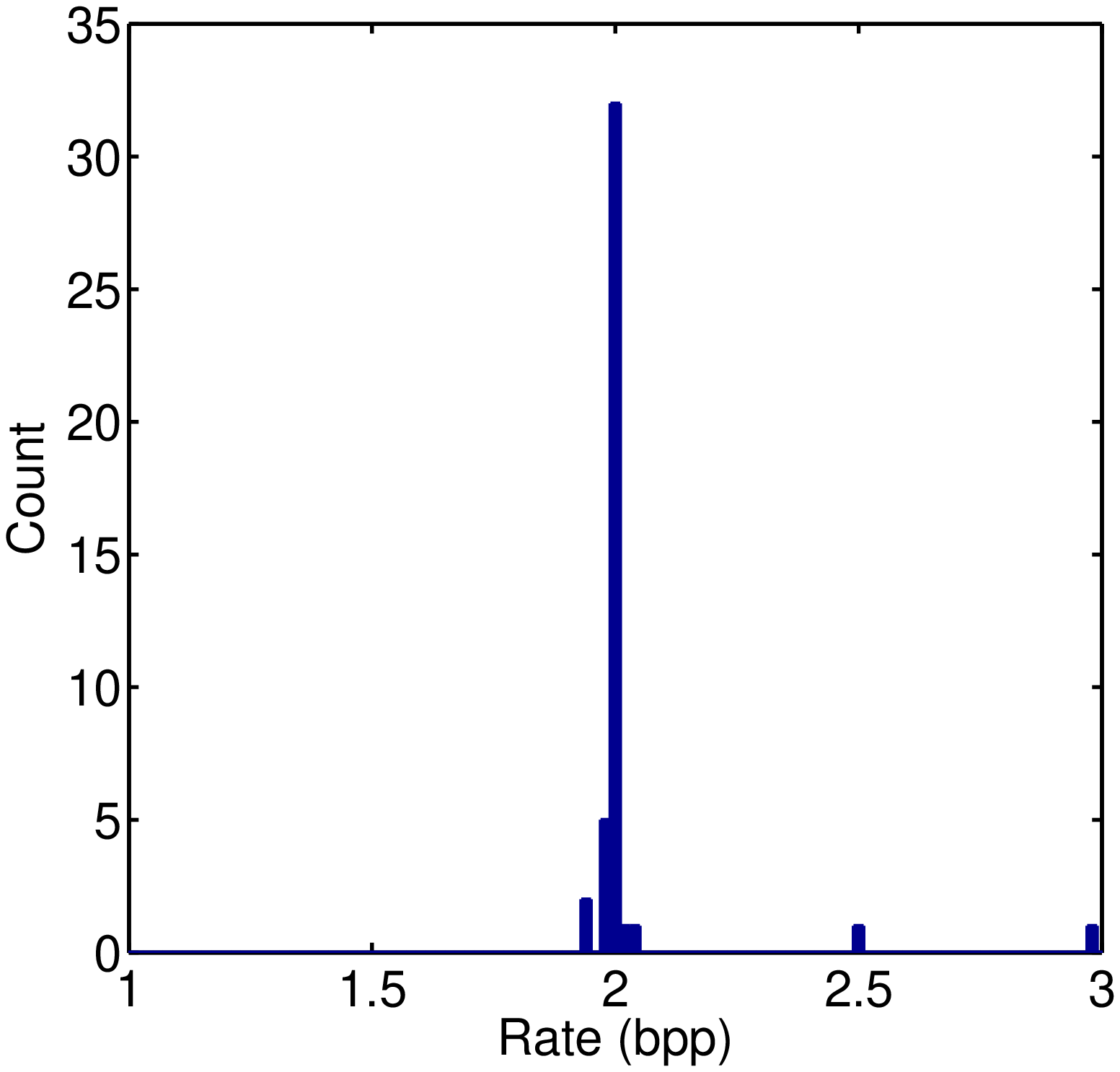}} 
\subfigure[3 bpp]{
\includegraphics[width=0.3\textwidth]{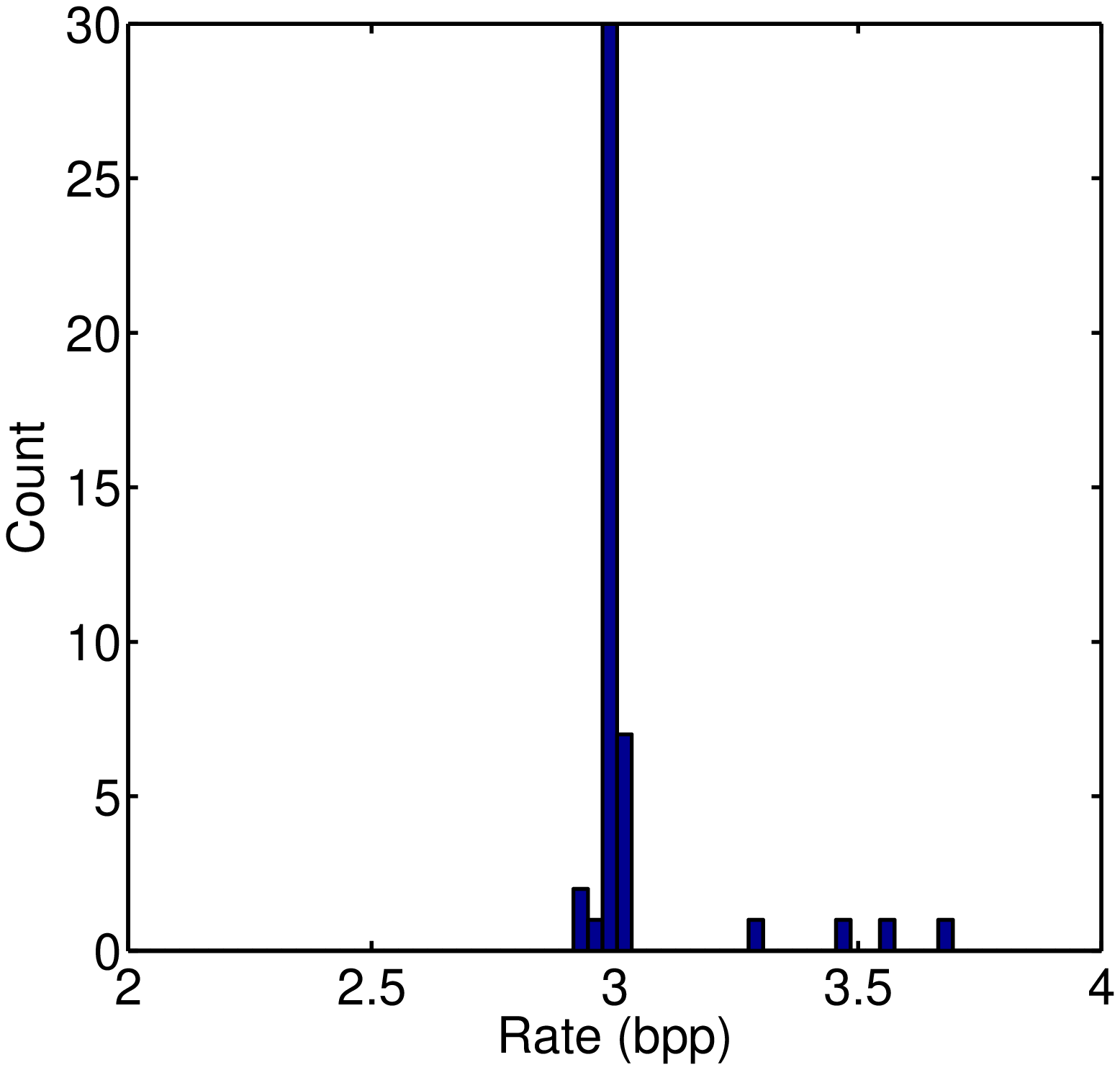}}
\vspace*{-0.2cm}
\caption{Histograms of output rates for mode B}
\label{histograms}
\end{figure*}

\begin{table*}
\caption{Output rates}
\centering
\begin{tabular}{c| c| c| c c c c}
Image & Size (lines$\times$pixels$\times$bands) & Mode & \emph{1 bpp} & \emph{2 bpp} & \emph{3 bpp} & \emph{4 bpp}\\
\hline
\multirow{2}{*}{\textsc{AVIRIS sc0\_raw}} & \multirow{2}{*}{$512\times680\times224$} & A & 0.951 & 1.963 & 2.955 & 3.959 \\
& & B & 1.004 & 1.995 & 3.006 & 3.994\\  
\hline
\multirow{2}{*}{\textsc{AIRS gran9}} & \multirow{2}{*}{$135\times90\times1501$} & A & 0.948 & 1.939 & 2.963 & 3.971 \\
& & B & 0.959 & 1.976 & 2.962 & 3.962\\ 
\hline
\multirow{2}{*}{\textsc{CASI-t0477f06-nuc}} & \multirow{2}{*}{$1225\times406\times72$} & A & 0.881 & 1.944 & 2.924 & 3.981 \\
& & B & 0.999 & 1.994 & 2.995 & 3.988\\ 
\hline
\multirow{2}{*}{\textsc{CRISM-sc167-nuc}} & \multirow{2}{*}{$510\times640\times545$} & A & 0.678 & 1.706 & 2.677 & 3.690 \\
& & B & 1.003 & 1.993 & 2.986 & 3.991 \\ 
\hline
\multirow{2}{*}{\textsc{CRISM-sc182-nuc}} & \multirow{2}{*}{$450\times320\times545$} & A & 0.680 & 1.698 & 2.691 & 3.696 \\
& & B & 1.002 & 1.991 & 2.985 & 3.973 \\ 
\hline
\multirow{2}{*}{\textsc{frt00009326\_07\_vnir}} & \multirow{2}{*}{$512\times640\times107$} & A & 0.607 & 1.427 & 2.231 & 3.140 \\
& & B & 1.000 & 1.999 & 3.001 & 3.994 \\ 
\hline
\multirow{2}{*}{\textsc{Geo\_Sample\_Flatfielded}} & \multirow{2}{*}{$1024\times256\times242$} & A & 0.912 & 2.070 & 3.124 & 3.998 \\
& &  B & 0.988 & 1.987 & 2.983 & 3.971 \\ 
\hline
\multirow{2}{*}{\textsc{M3targetB-nuc}} & \multirow{2}{*}{$512\times640\times260$} & A & 0.889 & 1.974 & 3.043 & $3.834^{(*)}$ \\
& & B & 1.000 & 1.997 & 2.998 & $3.834^{(*)}$ \\ 
\hline
\multirow{2}{*}{\textsc{MODIS-MOD01\_250m}} & \multirow{2}{*}{$8120\times5416\times2$} & A & 0.909 & 1.997 & 2.939 & 3.839 \\
& & B & 1.014 & 2.009 & 3.006 & 4.004 \\ 
\hline
\multirow{2}{*}{\textsc{MODIS-MOD01day}} & \multirow{2}{*}{$2030\times1354\times14$} & A & 1.042 & 2.045 & 2.996 & 3.985 \\
& & B & 1.014 & 2.005 & 2.998 & 3.986 \\ 
\hline
\multirow{2}{*}{\textsc{montpellier}} & \multirow{2}{*}{$224\times2456\times4$} & A & 0.959 & 2.122 & 3.123 & 4.105 \\
& & B & 1.025 & 2.035 & 3.030 & 4.032 \\ 
\hline
\multirow{2}{*}{\textsc{mountain}} & \multirow{2}{*}{$1024\times1024\times6$} & A & 0.735 & 1.935 & 2.970 & $3.793^{(*)}$ \\
& & B & 1.002 & 2.003 & 3.003 & $3.793^{(*)}$ \\ 
\hline
\multirow{2}{*}{\textsc{t0477f06\_raw}} & \multirow{2}{*}{$1225\times406\times72$} & A & 1.138 & 1.971 & 2.935 & 3.987 \\
& & B & 1.016 & 1.994 & 2.995 & 3.993 \\ 
\hline
\multirow{2}{*}{\textsc{toulouse\_spot5\_xs\_extract1}} & \multirow{2}{*}{$1024\times1024\times3$} & A & 0.714 & 1.815 & 2.805 & 3.802 \\
& & B & 1.010 & 2.002 & 2.999 & 3.997 \\ 
\hline
\multirow{2}{*}{\textsc{vgt\_1b}} & \multirow{2}{*}{$10080\times1728\times4$} & A & 0.630 & 1.813 & 2.878 & 3.914 \\
& & B & 1.009 & 2.004 & 3.002 & 4.001 \\ 
\hline
\end{tabular}
\\
\vspace*{0.4cm}
(*) : lossless
\label{tab:rates}
\end{table*}

In this section we show some results concerning the accuracy of the rate controller in terms of output rate. The tests are conducted for various target rates, and for the two operating modes of the algorithm: A and B. The predictor defined in the CCSDS-123 standard is used in the full prediction mode and with neighbour-oriented local sums. Square blocks of size $16\times16$ are used but the variance of the unquantized prediction residuals is obtained by running the lossless encoder on 2 lines only. This allows to buffer only two spectral lines at any given time, avoiding the need of large onboard memory buffers. Table \ref{tab:rates} reports a selection of the test images and the output rates obtained for the specified target rates. While later we will report full rate-distortion results, this test aims at assessing the accuracy achieved at obtaining a given target rate. It can be noted that the operating mode A is typically less accurate than mode B. Nonetheless it can still get very good accuracy in many cases, and, as explained in Sec. \ref{RDnumericalperformance}, it potentially has better rate-distortion characteristics. Mode B always has remarkably good accuracy, thanks to the information on the actual number of bits used to encode previous slices. Moreover, it can be seen that the algorithm performs equally well on both hyperspectral and multispectral images.

Furthermore, Fig. \ref{histograms} reports histograms of the actual rate obtained by mode B on a total of 47 images belonging to the test set of the CCSDS. The bin width is 1\% of the target rate. It should be noticed that, for the histogram at 3 bpp, some of the images were encoded without losses using a rate lower than the target, hence they have not been considered in the histograms. Notice that many images in the test set reach accuracy as good as 1\% or less. We remark that the rate control results are consistent throughout this large test set and only few images failed to be encoded with good accuracy. This is due to the severe noise affecting those images, causing the predictor to have low performance, and consequently, the prediction residuals exhibit large deviations from the model we assumed.

\subsection{Rate-distortion performance}
\label{RDnumericalperformance}
\begin{figure*}[ht!]
\subfigure[]{
\includegraphics[width=0.32\textwidth]{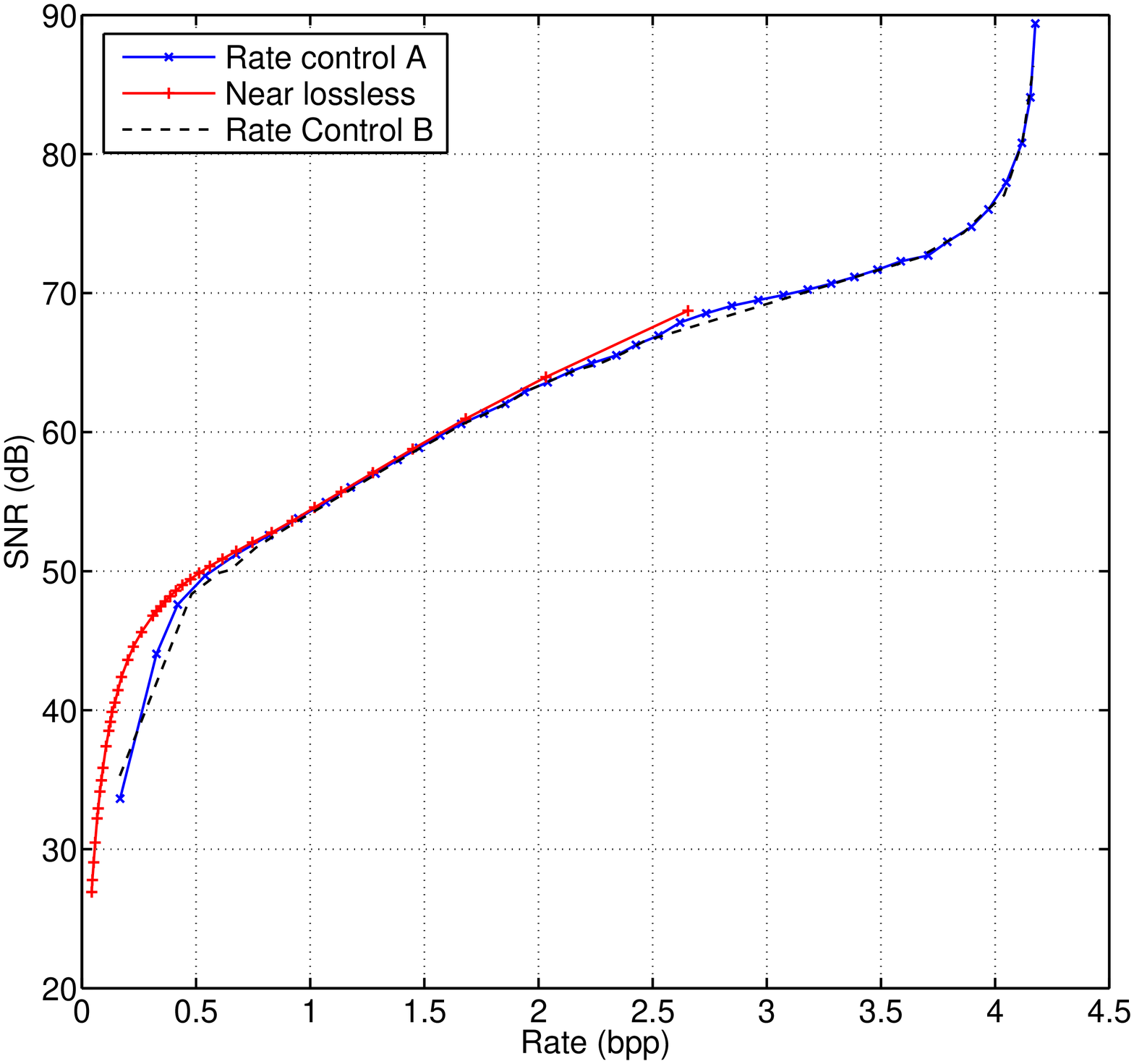}} 
\subfigure[]{
\includegraphics[width=0.32\textwidth]{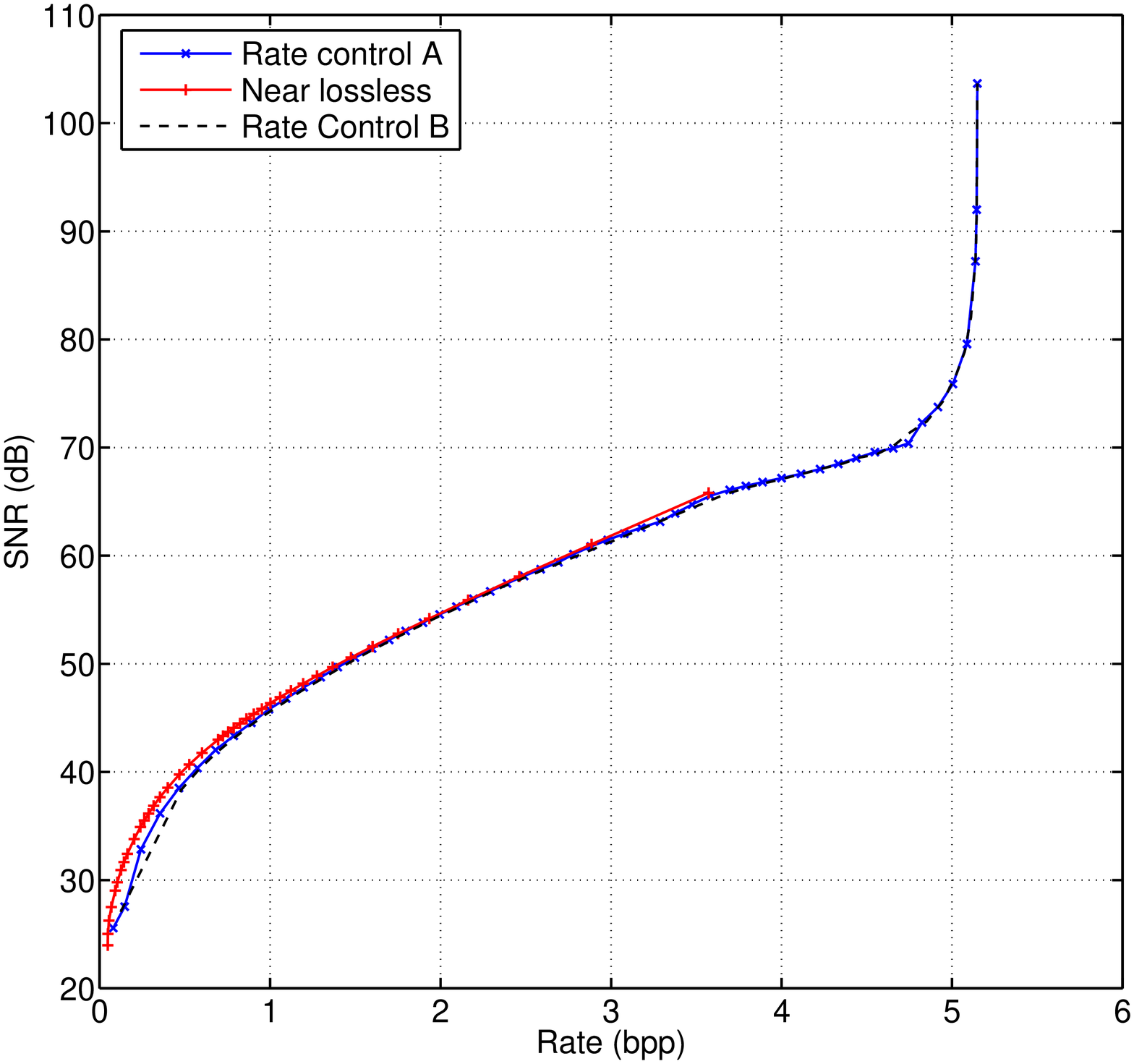}} 
\subfigure[]{
\includegraphics[width=0.32\textwidth]{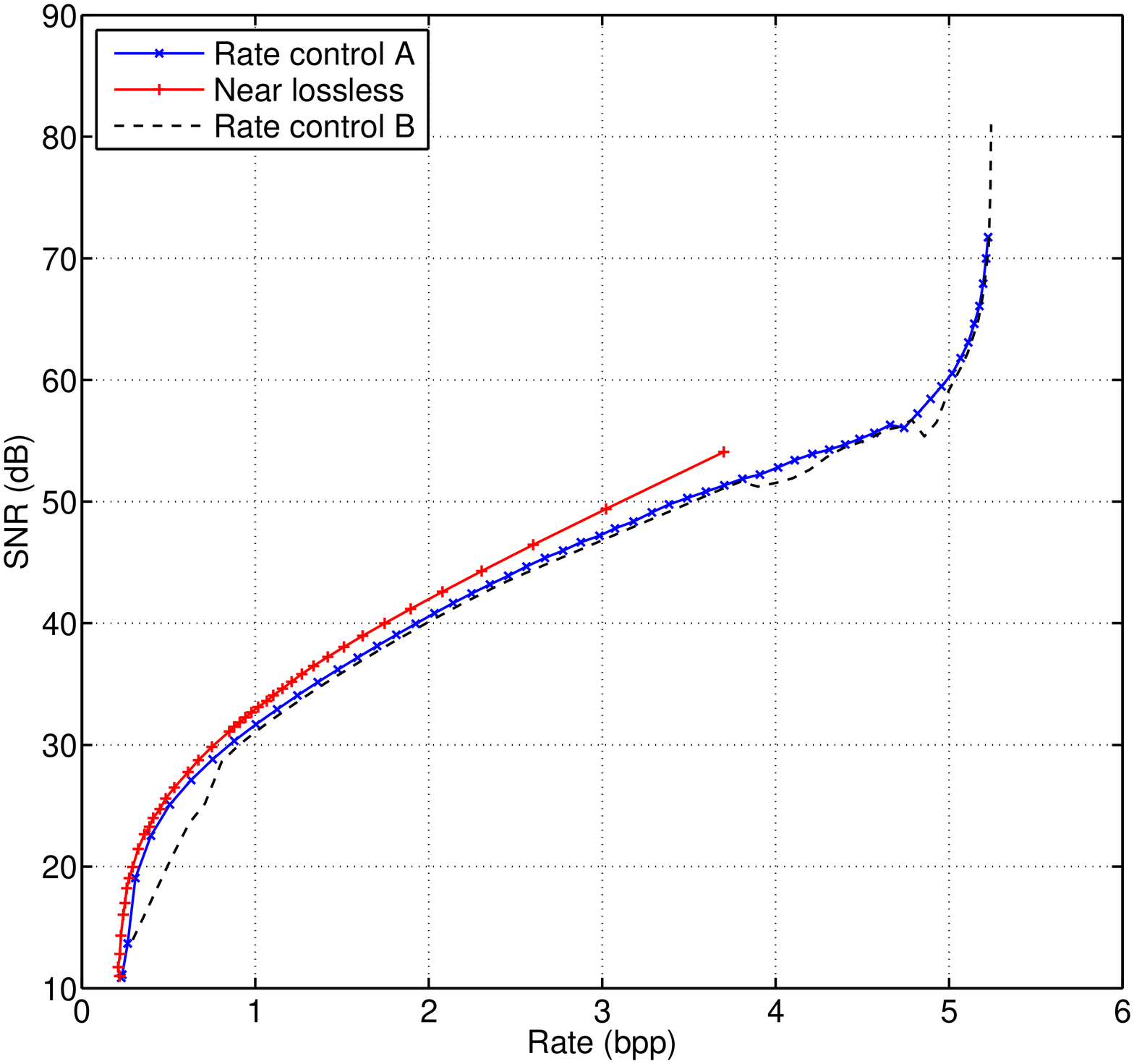}} 
\caption{Rate-SNR curves. (a) \emph{AIRS gran9} , (b) \emph{CRISM-sc182-nuc}, (c) \emph{vgt1\_1b}}
\label{rateditortion}
\end{figure*}
In this section we study the rate-distortion performance of the encoder, and, in particular, we focus on the suboptimality of the rate controller with respect to a near-lossless encoding of the images. The problem with near-lossless compression is that, apart from the lack of rate control, only certain rates can be achieved due to choice of a single quantization step for the whole image. At high rates, this causes rate-distortion points to be quite far apart from each other (\emph{e.g}, as much as $0.5$ bpp), hence not allowing very flexible choices for the rate-distortion operating point. On the other hand, rate control allows to achieve very fine granularity and any rate-distortion point, from low rates up to lossless compression, can be used.
Figure \ref{rateditortion} shows the rate-SNR curves obtained for near-lossless compression, rate control with mode A and rate control with mode B for some test images. The following definition of SNR is used throughout the paper: 
\begin{align*}
\textsc{SNR} = 10\log_{10} \frac{\sum_{i=1}^{N_{pixels}}x_i^2}{\sum_{i=1}^{N_{pixels}}(x_i-\hat{x}_i)^2}
\end{align*}
being $x_i$ and $\hat{x}_i$ the $i$-th pixel in the original image and in the decoded image, respectively. As already explained in section \ref{sec:modeB}, the great accuracy in the rate achieved by mode B is paid in terms of slightly lower rate-distortion performance. However, it is remarked that when the encoder is run relying on the rate control only, the greater accuracy of mode B often results in better quality than that provided by mode A, which often yields a rate lower than the target. Nevertheless, it can be noticed that the rate-distortion curves for both mode A and mode B are quite close to the near-lossless performance. As an example, for \emph{AIRS gran9} the gap is only about $0.2$ dB at 2 bpp. For \emph{frt00009326\_07\_vnir} the gap at 2 bpp is 0.2 dB for mode A and 0.4 bpp for mode B. We report image \emph{vgt1\_1b} as one of the worst cases of rate-SNR performance, where mode A loses about 1.5 dB with respect to near-lossless encoding and mode B about 1.8 dB, always at 2 bpp.   
We also remark that the curves were obtained without constraining maximum distortion, which can significantly improve performance, as shown in Sec. \ref{hybrid}.

\subsection{Comparison with transform coding}

The CCSDS-122 standard \cite{ccsds122} defines a transform coder employing the Discrete Wavelet Transform and a low-complexity Bit Plane Encoder, for the compression of 2D imagery. An extension of such standard to multiband imagery by including a spectral transform has been implemented and is publicly available online \cite{deltasoftware}. The implementation combines the CCSDS-122 encoder with the POT spectral transform \cite{pot}. The proposed system is run using the memory-1 mode B of rate control (slice-by-slice feedback) with $\tau=5$, with full prediction mode and neighbor-oriented local sums, while the transform system performs the rate allocation by means of the reverse waterfill algorithm \cite{taubmanjpeg2000}. We remark that the availability of the rate controller for the predictive system allows to perform a direct comparison, in which both systems work in a pure rate-controlled fashion by specifying a target rate and letting the encoder perform all the coding decisions automatically. The proposed rate controller is operated using $16\times 16$ blocks and $\mathrm{ESTLINES}=2$, meaning that only two lines out of 16 are used for estimation of the variance of unquantized prediction residuals. On the other hand, the transform coding system buffers 8 lines, thus requiring more memory. 
Table \ref{hydra_vs_delta} reports a comparison between the two systems, highlighting in bold the best results. 
The proposed predictive system is competitive against transform coding by typically providing superior quality, both in terms of SNR and in terms of maximum absolute distortion (MAD), for the same rate. Other quality metrics such as the maximum spectral angle (MSA) and average spectral angle (ASA) have been studied in the literature \cite{quality_metrics}, but we omit them for reasons of brevity. However, such metrics follow the same trends observed for SNR and MAD, respectively. We observe that, at lower rates, the proposed algorithm achieves significant gains in terms of MAD even when the SNR gain is small or for the few cases when the transform coder is more effective.
We also report (Table \ref{hydra_gains}) the mean and median gains in terms of SNR and MAD obtained by the proposed algorithm on the whole corpus of images. We choose to report the median gain, as well as the mean, due to some outliers in the results that bias the mean gain statistics due to the large gain obtained by the proposed system. It is sometimes the case that the proposed system reaches lossless quality for the desired rate, while the transform coder does not. Such cases are excluded from the computation of the SNR gain as it would be infinite. We can notice that the higher gains are achieved for higher rates, confirming the typical behaviour of predictive encoders with respect to transform encoders.
Finally, we report a visual comparison (Fig. \ref{visual_comp}) on a cropped portion of the first band of the \emph{vgt1\_1b} test image. The two algorithms are compared at the same rate of 2 bpp. Although it is difficult to see the differences with the naked eye on paper, the figures reporting the magnitude of the error clearly show that the proposed predictive approach consistently achieves smaller deviations from the original image. Also, notice that despite the block-based approach of the proposed algorithm, scalar quantization of the prediction residuals does not produce blocketization artifacts.

\begin{figure*}
\centerline{
\includegraphics[width=0.185\textwidth]{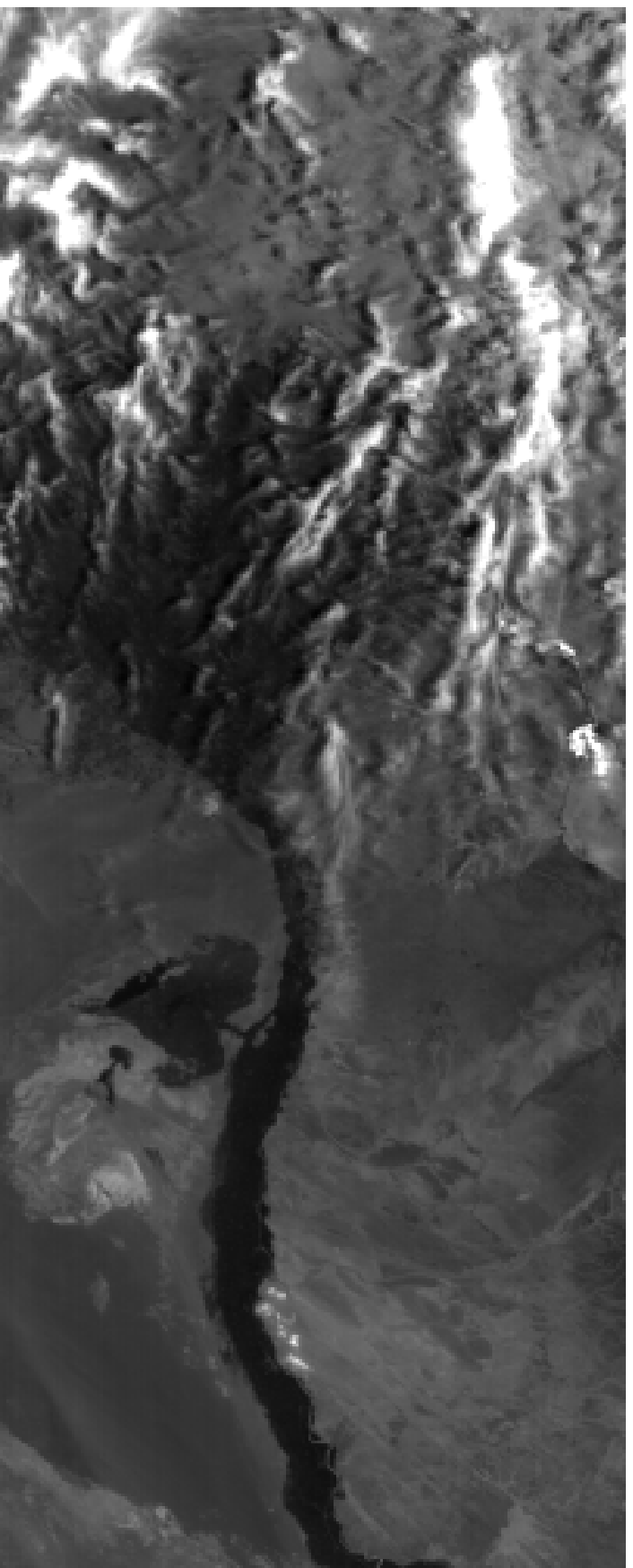}
\includegraphics[width=0.185\textwidth]{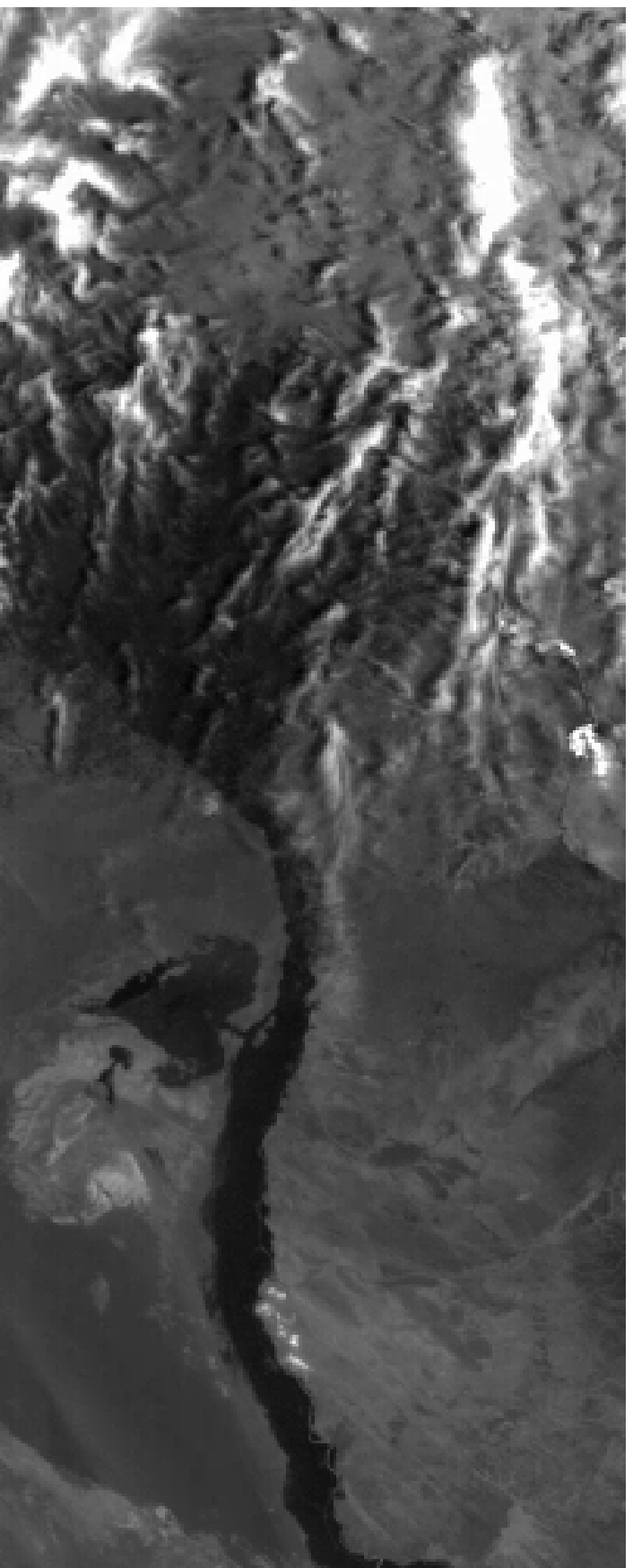}
\includegraphics[width=0.185\textwidth]{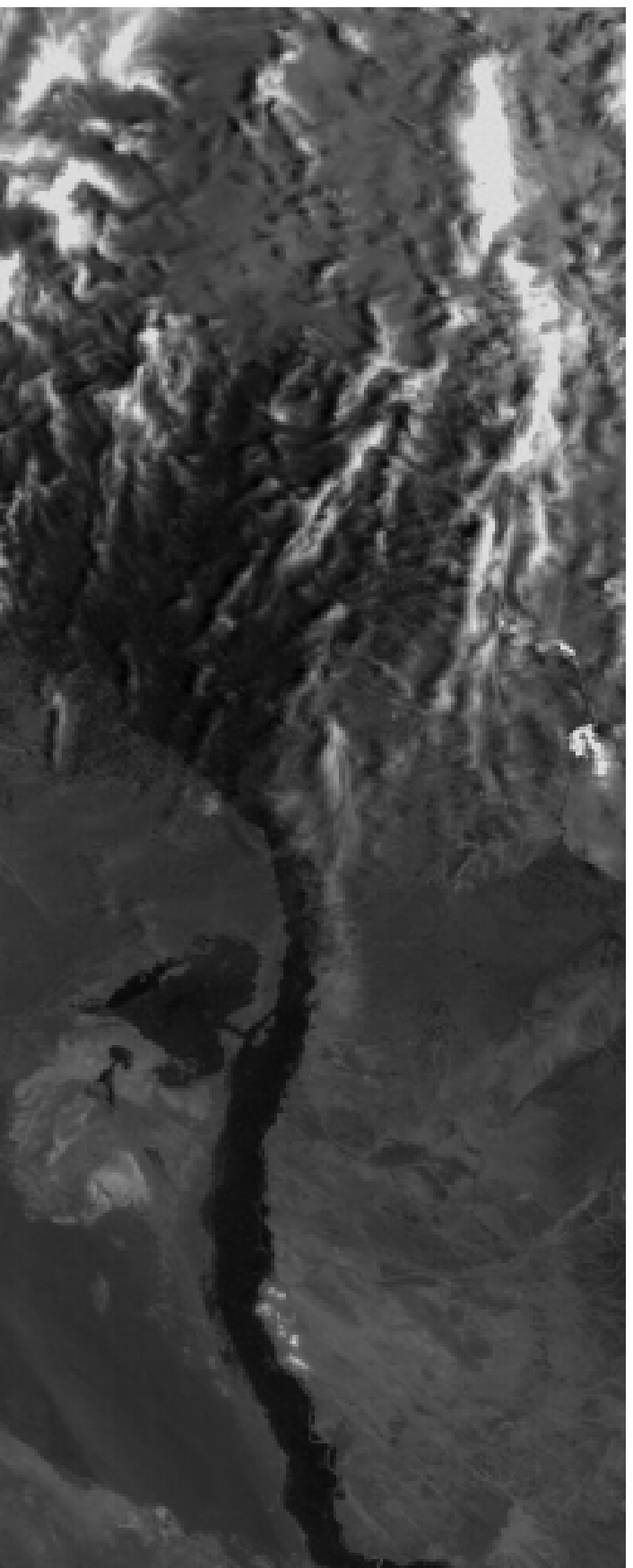}
\includegraphics[width=0.185\textwidth]{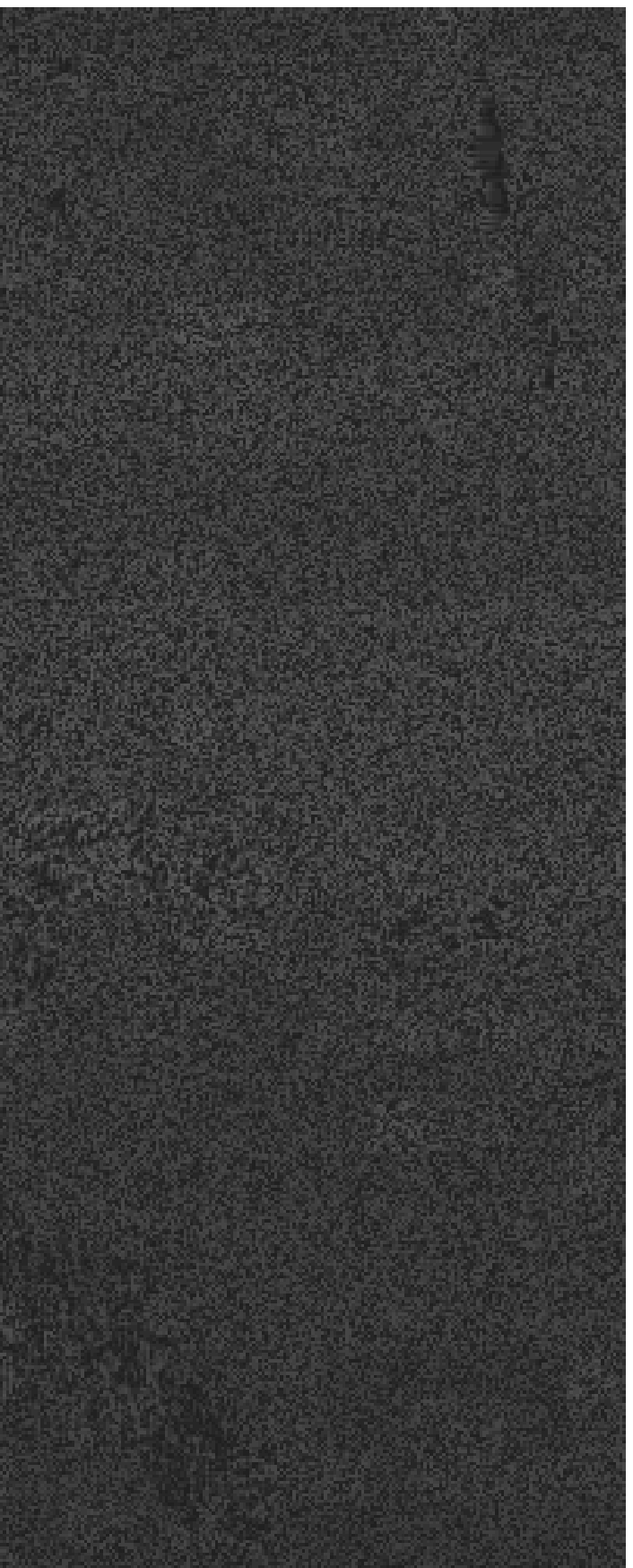}
\includegraphics[width=0.185\textwidth]{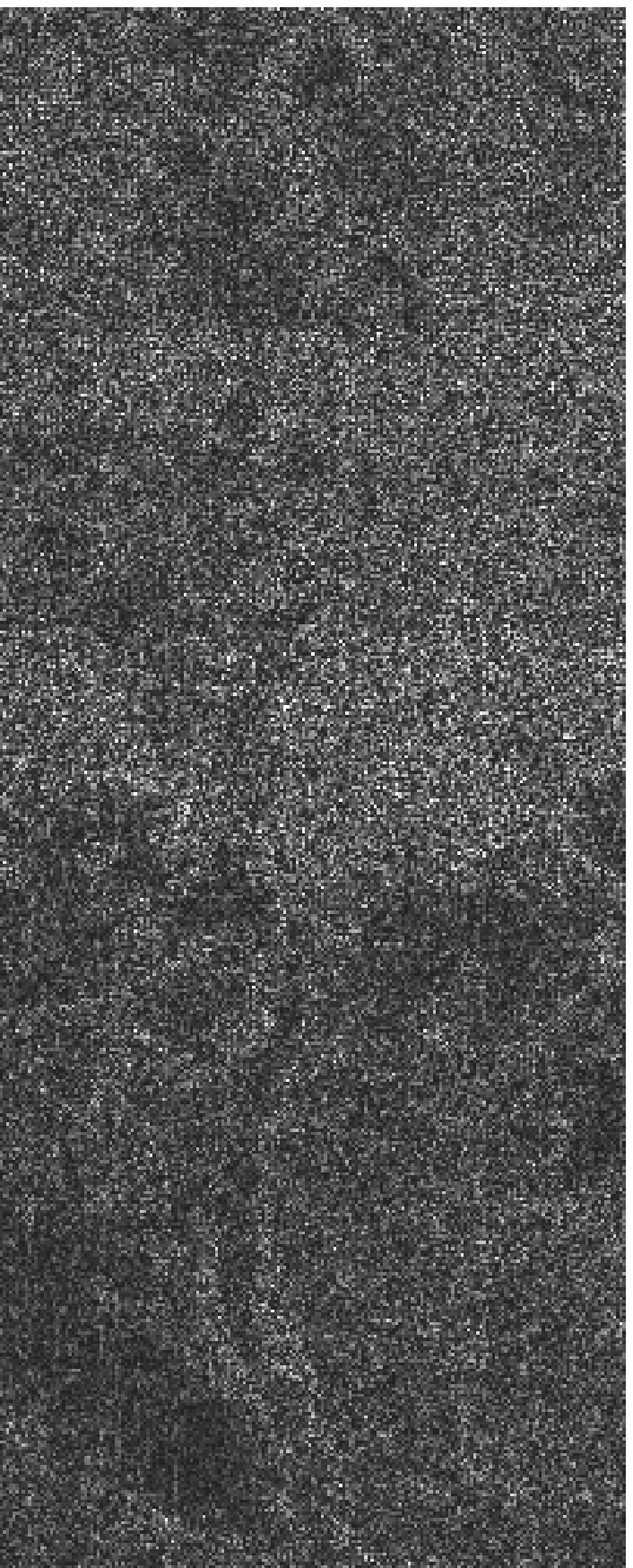}}
\caption{Visual comparison of a crop of \emph{vgt1\_1b}, band 1. From left to right: original image, predictive approach, transform approach, absolute error for predictive, absolute error for transform}
\label{visual_comp}
\end{figure*}

\vspace*{-0.2cm}
\section{Conclusions}
In this paper we have presented a rate control algorithm for onboard compression of hyperspectral and multispectral images designed to work with predictive encoders and suitable for implementation on spacecrafts. While rate control is easy to perform in the case of transform coding, the predictive coding paradigm poses significant challenges. We have proposed a scheme based on modelling the predicted rate and distortion for non-overlapping blocks of the image and optimizing the assignment of quantization step sizes over slices of the image. Extensive tests have shown that the algorithm can effectively control the output rate with excellent accuracy. Moreover, rate control solves one of the issues of near-lossless compression, \emph{i.e.}, the scarce number of operating points at high rates. In fact, the availability of a rate controller allows the user to choose any rate, depending on their specific needs. We have also proposed an extension of the CCSDS-123 standard to deal with lossy, near-lossless and hybrid near-lossless rate-controlled compression in a single package. The resulting architecture is competitive with the transform coding approach, significantly outperforming it at all rates from 1 bpp up to lossless compression.

\begin{table*}[htbp]
\caption{PREDICTIVE (CCSDS-123 + Rate Control B) vs. TRANSFORM (CCSDS-122 + POT + Reverse Waterfill)}
\centerline{
\begin{tabular}{c|c|ccc|cc}
 \multicolumn{ 2}{c}{} & \multicolumn{ 3}{c}{PREDICTIVE} & \multicolumn{ 2}{c}{TRANSFORM} \\ 
IMAGE & TARGET (bpp) & RATE (bpp) & SNR (dB) & MAD & SNR (dB) & MAD \\ 
\hline
 & 1.00 & 1.004 & 46.15 & \bf{247} & \bf{48.87} & 433 \\ 
\textsc{aviris\_sc0} & 2.00 & 1.995 & \bf{55.93} & \bf{35} & 55.02 & 107 \\ 
$512\times680\times224$ & 3.00 & 3.006 & \bf{62.28} & \bf{30} & 59.69 & 41 \\ 
 & 4.00 & 3.994 & \bf{69.49} & \bf{3} & 65.03 & 21 \\ 
 \hline
 & 1.00 & 0.951 & \bf{47.79} & \bf{10} & 44.24 & 167 \\ 
\textsc{CRISM-sc214-nuc} & 2.00 & 1.938 & \bf{56.26} & \bf{4} & 52.72 & 45 \\ 
$510\times640\times545$ & 3.00 & 2.920 & \bf{62.09} & \bf{2} & 60.33 & 7 \\ 
 & 4.00 & 3.818 & \bf{88.15} & \bf{1} & 65.32 & 3 \\ 
 \hline
 & 1.00 & 1.001 & \bf{44.00} & \bf{63} & 37.16 & 1278 \\ 
\textsc{M3targetB} & 2.00 & 1.998 & \bf{54.25} & \bf{10} & 46.38 & 243 \\ 
$512\times640\times260$ & 3.00 & 2.997 & \bf{60.08} & \bf{8} & 58.52 & 32 \\ 
 & 4.00 & 3.929 & \bf{69.55} & \bf{2} & 64.88 & 7 \\ 
 \hline
 & 1.00 & 1.016 & \bf{29.36} & \bf{255} & 25.87 & 752 \\ 
\textsc{MODIS-MOD01\_500m} & 2.00 & 2.009 & \bf{39.32} & \bf{100} & 36.54 & 244 \\ 
$4060\times2708\times5$ & 3.00 & 3.005 & \bf{47.25} & \bf{37} & 42.99 & 168 \\ 
 & 4.00 & 4.002 & \bf{54.08} & \bf{16} & 49.77 & 53 \\ 
 \hline
 & 1.00 & 1.004 & 42.19 & \bf{88} & \bf{43.60} & 618 \\ 
\textsc{MODIS-MOD01night} & 2.00 & 2.002 & \bf{52.49} & \bf{32} & 51.10 & 276 \\ 
$2030\times1354\times17$ & 3.00 & 3.001 & \bf{59.84} & \bf{12} & 56.51 & 277 \\ 
 & 4.00 & 4.000 & \bf{65.78} & \bf{5} & 61.48 & 47 \\ 
 \hline
 & 1.00 & 1.024 & \bf{28.67} & \bf{137} & 27.16 & 670 \\ 
\textsc{montpellier} & 2.00 & 2.035 & \bf{37.23} & \bf{42} & 33.46 & 635 \\ 
$224\times2456\times4$ & 3.00 & 3.030 & \bf{44.60} & \bf{18} & 39.88 & 92 \\ 
 & 4.00 & 4.032 & \bf{51.22} & \bf{7} & 45.44 & 47 \\ 
 \hline
 & 1.00 & 1.010 & 23.80 & \bf{25} & \bf{24.26} & 74 \\ 
\textsc{toulouse\_spot5\_xs\_extract1} & 2.00 & 2.002 & \bf{31.88} & \bf{7} & 30.37 & 36 \\ 
$1024\times1024\times3$ & 3.00 & 2.999 & \bf{38.53} & \bf{4} & 35.98 & 14 \\ 
 & 4.00 & 3.997 & \bf{44.30} & \bf{2} & 41.29 & 7 \\ 
 \hline
 & 1.00 & 1.009 & \bf{31.07} & \bf{83} & 27.93 & 372 \\ 
\textsc{vgt1\_1b} & 2.00 & 2.004 & \bf{40.18} & \bf{27} & 37.05 & 231 \\ 
$10080\times1728\times4$ & 3.00 & 3.002 & \bf{47.31} & \bf{11} & 44.02 & 64 \\ 
 & 4.00 & 4.001 & \bf{53.52} & \bf{5} & 49.76 & 15 \\ 
 \hline
\end{tabular}}
\vspace*{-0.0cm}
\label{hydra_vs_delta}
\end{table*}

\begin{table*}[htbp]
\caption{MEAN AND MEDIAN GAIN}
\centerline{
\begin{tabular}{c| c c| c c }
RATE (bpp) & MEAN SNR GAIN (dB) & MEDIAN SNR GAIN (dB) & MEAN MAD GAIN & MEDIAN MAD GAIN \\
\hline
1.00 & 1.55 & 1.51 & 727 & 404 \\
\hline
2.00 & 2.82 & 2.17 & 323 & 117 \\
\hline
3.00 & 3.46 & 2.93 & 123 & 29 \\
\hline
4.00 & 6.6 & 4.31 & 54 & 7 \\
\hline
\end{tabular}}
\label{hydra_gains}
\end{table*}

\section{Acknowledgements}
We would like to thank Ian Blanes from the Universitat Aut\`{o}noma de Barcelona for precious support on using the Delta software developed by the Group on Interactive Coding of Images. 


\begin{IEEEbiography}[{\includegraphics[width=1in,height=1.25in,clip,keepaspectratio]{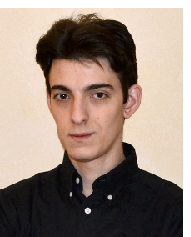}}]{Diego Valsesia}
Diego Valsesia received his MSc in Telecommunications Engineering from Politecnico di Torino and the MSc in Electrical and Computer Engineering from the University of Illinois at Chicago, both in 2012. He is currently a Ph.D. student at the Department of Electronics and Telecommunications of Politecnico di Torino. His main research interests include compression of remote sensing images, compressed sensing and sparse representations.
\end{IEEEbiography}
\vspace*{-10pt}
\begin{IEEEbiography}[{\includegraphics[width=1in,height=1.25in,clip,keepaspectratio]{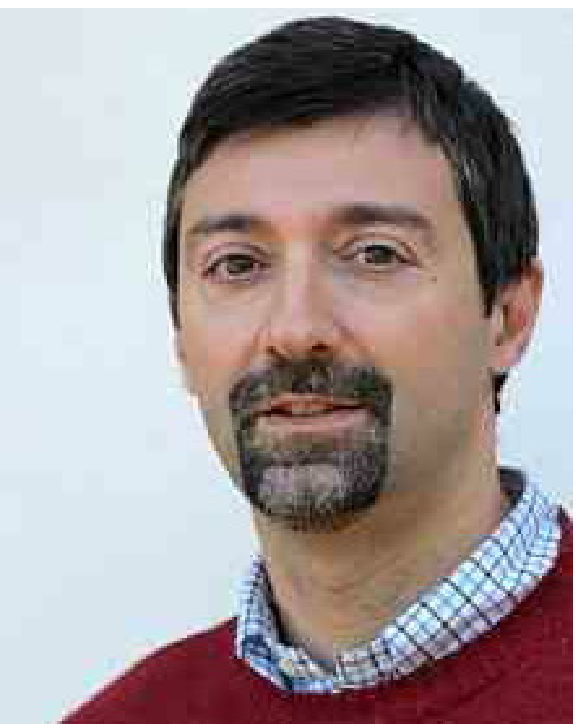}}]{Enrico Magli}
Enrico Magli received the Ph.D. degree in Electrical Engineering in 2001, from Politecnico di Torino, Italy. He is currently an Associate Professor at the same university, where he leads the Image Processing Lab. His research interests are in the field of compression of satellite images, multimedia signal processing and networking, compressive sensing, distributed source coding, image and video security. He is an associate editor of the IEEE Transactions on Circuits and Systems for Video Technology, and of the EURASIP Journal on Information Security. He has been co-editor of JPEG 2000 Part 11 - Wireless JPEG 2000. He is general co-chair of IEEE MMSP 2013, and has been TPC co-chair of ICME 2012, VCIP 2012, MMSP 2011 and IMAP 2007. He has published about 40 papers in refereed international journals, 3 book chapters, and over 90 conference papers. He is a co-recipient of the IEEE Geoscience and Remote Sensing Society 2011 Transactions Prize Paper Award, and has received the 2010 Best Reviewer Award of IEEE Journal of Selected Topics in Applied Earth Observation and Remote Sensing.  He has received a 2010 Best Associate Editor Award of IEEE Transactions on Circuits and Systems for Video Technology. 
\end{IEEEbiography}

\end{document}